\documentclass[]{article}

\usepackage[letterpaper]{geometry}

\usepackage[T1]{fontenc}
\usepackage{amsfonts}
\usepackage{graphicx}
\usepackage{epstopdf}
\usepackage{enumitem}
\usepackage{algorithm}
\ifpdf
  \DeclareGraphicsExtensions{.eps,.pdf,.png,.jpg}
\else
  \DeclareGraphicsExtensions{.eps}
\fi

\usepackage{amsopn}

\usepackage{amsmath,amsfonts,amssymb}
\usepackage{xifthen}
\usepackage{bm}
\usepackage[colorlinks=true]{hyperref}
\usepackage{cleveref}
\usepackage{tikz}
\usepackage{etoolbox}
\usepackage{mathrsfs}

\usepackage{algpseudocode}

\usepackage{amsthm}

\theoremstyle{definition}
\newtheorem{theorem}{Theorem}[section]
\newtheorem{lemma}[theorem]{Lemma}
\newtheorem{corollary}[theorem]{Corollary}
\newtheorem{definition}[theorem]{Definition}

\newcommand{\set}[1]{\left\{#1\right\}}

\newcommand{\eps}{\varepsilon}
\newcommand{\inner}[2]{\langle #1,#2\rangle}

\renewcommand{\mid}{\;\middle\vert\;}

\newcommand{\defeq}{:=}

\renewcommand{\emptyset}{\varnothing}
\renewcommand{\Pr}[2][]{ \ifthenelse{\isempty{#1}}
  {\mathbf{Pr}\left[#2\right]} {\mathbf{Pr}_{#1}\left[#2\right]} }
\newcommand{\E}[2][]{ \ifthenelse{\isempty{#1}}
  {\mathbf{E}\left[#2\right]}
  {\mathbf{E}_{#1}\left[#2\right]} }
\newcommand{\Var}[2][]{ \ifthenelse{\isempty{#1}}
  {\mathbf{Var}\left[#2\right]}
  {\mathbf{Var}_{#1}\left[#2\right]} }
\newcommand{\Ent}[2][]{ \ifthenelse{\isempty{#1}}
  {\mathbf{Ent}\left[#2\right]}
  {\mathbf{Ent}_{#1}\left[#2\right]} }

  \newcommand{\W}[1]{W_1\left( #1 \right)}
  \newcommand{\Wd}[2][d]{W^{#1}_1\left( #2 \right)}
  \newcommand{\Wv}[1]{\Wd[V]{ #1 }}
  \newcommand{\We}[1]{\Wd[E]{ #1 }}

\title{Low-Sensitivity Matching via Sampling from Gibbs Distributions}
\author{Yuichi Yoshida\thanks{National Institute of Informatics,  supported by JSPS KAKENHI Grant Number 24K02903, \texttt{yyoshida@nii.ac.jp}.}
   \and Zihan Zhang\thanks{National Institute of Informatics, supported by JST SPRING Grant Number JPMJSP2104, \texttt{zihan@nii.ac.jp}}}
\date{}

\begin{document}

\maketitle

\begin{abstract}
    In this work, we study the maximum matching problem from the perspective of sensitivity. 
    The sensitivity of an algorithm $A$ on a graph $G$ is defined as the maximum Wasserstein distance between the output distributions of $A$ on $G$ and on $G - e$, where $G - e$ is the graph obtained by deleting an edge $e$ from $G$. 
    The maximum is taken over all edges $e$, and the underlying metric for the Wasserstein distance is the Hamming distance.
    
    We first show that for any $\varepsilon > 0$, there exists a polynomial-time $(1 - \varepsilon)$-approximation algorithm with sensitivity $\Delta^{O(1/\varepsilon)}$, where $\Delta$ is the maximum degree of the input graph. 
    The algorithm is based on sampling from the Gibbs distribution over matchings and runs in time $O_{\varepsilon, \Delta}(m \log m)$, where $m$ is the number of edges in the graph.
    This result significantly improves the previously known sensitivity bounds.
    
    Next, we present significantly faster algorithms for planar and bipartite graphs as a function of $\varepsilon$ and $\Delta$, which run in time $\mathrm{poly}(n/\varepsilon)$.
    This improvement is achieved by designing a more efficient algorithm for sampling matchings from the Gibbs distribution in these graph classes, which improves upon the previous best in terms of running time.
    
    Finally, for general graphs with potentially unbounded maximum degree, we show that there exists a polynomial-time $(1 - \varepsilon)$-approximation algorithm with sensitivity $\sqrt{n} \cdot (\varepsilon^{-1} \log n)^{O(1/\varepsilon)}$, improving upon the previous best bound of $O(n^{1/(1+\varepsilon^2)})$.
\end{abstract}

\section{Introduction}\label{sec:intro}

In this work, we study the sensitivity of the maximum matching problem, which is recently proposed by Varma and Yoshida~\cite{varma2023average}.
Roughly speaking, the sensitivity of an algorithm quantifies how much its output changes, measured in Hamming distance, when a single edge is removed from the input graph.
For randomized algorithms, sensitivity is defined using the Wasserstein distance, where the underlying metric between outputs is the Hamming distance (see \Cref{sec:prelim} for details).
Note that the trivial upper bound on the sensitivity for the maximum matching problem is $O(n)$, where $n$ is the number of vertices, because the matching size is $O(n)$.
Although sensitivity is an intriguing concept on its own, it is tightly linked to differential privacy~\cite{dwork2006differential} and the \emph{non-signaling} model \cite{akbari2025online,arfaoui2014can,coiteux2024no,gavoille2009can} (see \Cref{subsec:related} for more details).
Moreover, low sensitivity is crucial for enabling trustworthy decision-making and reliable knowledge discovery.
See~\cite{varma2023average} for further discussion, and \cite{hara2023average,li2025average,peng2020average,yoshida2022average} for algorithms with small average sensitivity, a slightly weaker notion of sensitivity, applied to data mining and machine learning problems.

Although maximum matchings can be computed in polynomial time, it is known that any such algorithm must have  sensitivity $\Omega(n)$, even on simple structures like paths~\cite{varma2023average}. 
On the other hand, there exist algorithms with low sensitivity that achieve approximate solutions: specifically, a $1/2$-approximation algorithm with sensitivity $O(1)$~\cite{censor2016optimal,varma2023average}.
Moreover, for graphs with maximum degree bounded by $\Delta$, the connection between sensitivity and distributed algorithms in the $\mathsf{LOCAL}$ model~\cite{varma2023average}, combined with the $O(\varepsilon^{-3} \log \Delta)$-round distributed algorithm from~\cite{harris2019distributed}, yields a $(1 - \varepsilon)$-approximation algorithm with sensitivity $\Delta^{O(\varepsilon^{-3} \log \Delta)}$.

In this work, we achieve a significant improvement in the sensitivity of a $(1 - \varepsilon)$-approximation algorithm for the maximum matching problem on bounded-degree graphs.
We use $n$ and $m$ to denote the number of vertices and edges, respectively, of the input graph.
Then, we show the following:
\begin{theorem}\label{thm:bounded-degree-algorithm}
    For any $0<\eps<1$, there is a randomized $(1-\eps)$-approximation algorithm for the maximum matching problem with sensitivity $\Delta^{O(1/\eps)}$ and running time
    \[
    \min
    \left\{
    \exp(\Delta^{O(1/\varepsilon)})m\log n +O(n),
    \frac{1}{\eps} \Delta^{O(1/\eps)} n^2m\log n
    \right\}.
    \]
\end{theorem}
Our algorithm draws inspiration from statistical physics. Specifically, we consider the Gibbs distribution over matchings, where a matching $M$ is sampled with probability proportional to $\lambda^{|M|}$, for a parameter $\lambda > 0$, which is a special case of the \emph{monomer-dimer} model. 
When $\lambda$ is large, the distribution tends to favor large matchings. 
\Cref{thm:bounded-degree-algorithm} is obtained by carefully tuning the parameter to $\lambda = \Delta^{\Theta(1/\varepsilon)}$.
We also complement this result by showing that achieving a $(1-\varepsilon)$-approximation requires setting $\lambda = 2^{\Omega(1/\varepsilon)}$. 
Hence, the exponential dependence on $1/\varepsilon$ in the sensitivity bound of~\Cref{thm:bounded-degree-algorithm} is inherent for algorithms based on the Gibbs distribution.

Although the sensitivity in \Cref{thm:bounded-degree-algorithm} is low, the algorithm’s running time is relatively high, as the best known methods for sampling matchings incur a polynomial dependence on $\lambda$.
To address this, we show that our algorithm can be significantly accelerated for graph classes where perfect matchings can be efficiently sampled.
As representative examples, we describe the results for planar graphs and bipartite graphs here.
\begin{theorem}\label{thm:planar-algorithm}
    For any $0<\eps<1$, there is a randomized $(1-\eps)$-approximation algorithm for the maximum matching problem on planar graphs with sensitivity $\Delta^{O(1/\eps)}$ and running time $O(\eps^{-1} n^5m\log\Delta)$.
\end{theorem}
In particular, the algorithm runs in polynomial time even on general planar graphs, where we have no restriction on vertex degree (though the sensitivity is still of the form $\Delta^{O(1/\varepsilon)}$).
As a byproduct, we obtain an algorithm that outputs a random matching whose distribution is $\delta$-close in total variation distance to the Gibbs distribution induced by parameter $\lambda$, with running time $O(n^5 (m \log \lambda + \log (1/\delta)))$. 
This improves upon previous work~\cite{alimohammadi2021fractionally}, which only states that the running time is polynomial, although we suspect it to be as high as $O(n^5m^6 \log(1/\delta))$ (independent of $\lambda$).

For the bipartite graph case, we show the following:
\begin{theorem}\label{thm:bipartite-algorithm}
    For any $0<\eps<1$, there is a randomized $(1-\eps)$-approximation algorithm for the maximum matching problem on bipartite graphs with sensitivity $\Delta^{O(1/\eps)}$ and running time $\tilde O(\eps^{-3} n^{16}m^3\log^3\Delta)$.\footnote{$\tilde O(\cdot)$ hides a polylogarithmic factor in $n$.}
\end{theorem}
As in the planar case, we obtain as a byproduct an algorithm that outputs a random matching whose distribution is $\delta$-close (in total variation distance) to the Gibbs distribution induced by parameter $\lambda$, with running time $O(n^{12}(m \log \lambda + \log(1/\delta)))$.  
To the best of our knowledge, no algorithm is known for sampling a matching from the Gibbs distribution more efficiently on bipartite graphs than in the general case, where the best known runtime is $O(nm \lambda (n(\log n + \log \lambda)) + \log(1/\delta))$~\cite[Proposition 12.4]{jerrum1996markov}.  
In contrast, our algorithm achieves exponentially better time complexity as a function of $\lambda$.
As we choose $\lambda=\Delta^{\Theta(1/\varepsilon)}$, this difference is crucial.

Next, we show that sublinear sensitivity can be achieved on general graphs by first sparsifying the graph into a low-degree one with low sensitivity, while approximately preserving the matching size, using the technique developed in~\cite{chen2025entropy}, and then applying the algorithm from \Cref{thm:bounded-degree-algorithm}:
\begin{theorem}\label{thm:unbounded-degree-algorithm}
    For any $0<\eps<1$, there is a randomized $(1-\varepsilon)$-approximation algorithm for the maximum matching problem with edge sensitivity $\sqrt{n} (\varepsilon^{-1} \log n)^{O(1/\varepsilon)}$ and running time $O(n^2 m (\varepsilon^{-1}\log n)^{O(1/\varepsilon)})$.
\end{theorem}
To see the significance of \Cref{thm:unbounded-degree-algorithm}, consider a weaker notion of \emph{average} sensitivity, where the sensitivity is averaged over all possible edge deletions. 
Note that average sensitivity is always at most the (worst-case) sensitivity.
For this measure, a bound of $O(n^{1/(1 + \varepsilon^2)})$ is known~\cite{varma2023average}, which is only slightly sublinear. 
In contrast, our result establishes a \emph{worst-case} sensitivity of $\sqrt{n} \cdot (\varepsilon^{-1} \log n)^{O(1/\varepsilon)}$, providing a substantial improvement.\footnote{Although~\cite{yoshida2021sensitivity} claims a sensitivity of $O(3^K)$, where $K = (1/\varepsilon)^{2^{O(1/\varepsilon)}}$, there is a flaw in their argument concerning the sensitivity of the randomized greedy algorithm under vertex deletion. We will elaborate on this in \Cref{sec:flaw}.}

\subsection{Technical Overview}
\paragraph{Bounded-Degree Graphs}

We first discuss the approximation ratio of the matching sampled from the Gibbs distribution.
Let $m_G(\lambda) = \sum_{k \geq 0} m_k \lambda^k$ be the \emph{matching generating polynomial} of a graph $G$, where $m_k$ denotes the number of matchings of size $k$ in $G$. The expected size of a matching drawn from the Gibbs distribution with parameter $\lambda$ is given by
\[
    \frac{\sum_M |M| \lambda^{|M|}}{\sum_M \lambda^{|M|}} = \frac{\lambda m_G'(\lambda)}{m_G(\lambda)},
\]
where $M$ runs over matchings of $G$ and $m_G'$ is the derivative of $m_G$.
Let $\lambda_1, \ldots, \lambda_\nu$ be the (negative) roots of $m_G$, where $\nu$ is the size of a maximum matching in $G$. 
Then, by a standard calculation, the expected matching size can be expressed as
\[
    \sum_{i=1}^\nu \frac{\lambda}{\lambda - \lambda_i}.
\]

Using the celebrated result of Heilmann and Lieb on the location of the roots of the matching polynomial in bounded-degree graphs~\cite{heilmann1972theory}, we can show that a $(1 - \varepsilon)$-fraction of the roots $\lambda_i$ satisfy $|\lambda_i| \leq \Delta^{O(1/\varepsilon)}$, where $\Delta$ is the maximum degree of $G$. Hence, by choosing $\lambda = \Delta^{\Theta(1/\varepsilon)} / \varepsilon = \Delta^{\Theta(1/\varepsilon)}$, we ensure that the expected matching size achieves a $(1 - \varepsilon)$-approximation to the maximum matching size.

To bound the sensitivity of the Gibbs distribution, we establish a recursion on the sensitivity upper bounds for graphs of different sizes.
Similar techniques have been used to bound the spectral independence of the Gibbs distribution on matchings, as well as on other Holant problems and edge colorings~\cite{chen2023near,chen2024fast,chen2025decay}.

Before explaining the details, we introduce some notation.
Let $\mu_{E; G}$ denote the Gibbs distribution on matchings of a graph $G=(V,E)$.
For any $F \subseteq E$, let $\mu_{F; G}$ be the marginal distribution of $\mu_{E; G}$ on $2^F$ (see \Cref{def:marginal-conditional} for detail).
For an edge $i \in E$, we write $i \gets +$ to indicate that $i$ is matched and $i \gets -$ to indicate that $i$ is unmatched.
We call such an event a \emph{pinning}.
It is not hard to show that bounding the Wasserstein distance $\W{\mu_{E;G}, \mu_{E;G-i}}$ for an edge $i \in E$ can be reduced to bounding $\W{\mu_{E-i;G}^{i \gets +}, \mu_{E-i;G}^{i \gets -}}$.

A key property of the Gibbs distribution that yields low sensitivity is that conditioning on a pinning is equivalent to performing certain edge deletions on the underlying graph.
Hence, to bound the Wasserstein distance $\W{\mu_{E-i;G}^{i\gets +}, \mu_{E-i;G}^{i\gets -}}$ between the two distributions $\mu_{E-i;G}^{i\gets +}$ and $\mu_{E-i;G}^{i\gets -}$, we expand the distributions by the law of total probability over the matched/unmatched edges in $N(i)$, which gives a weighted sum of Wasserstein distances on distributions conditioned on pinnings of $N(i)$, where $N(i)$ is the set of edges incident to $i$.
Because of the equivalence of pinning and edge deletions, each of these conditional distributions is simply a Gibbs distribution on a smaller graph.
Thus, the Wasserstein distance between pinnings of the Gibbs distribution on $G$ can be controlled by a weighted sum of Wasserstein distances on those smaller graphs.
If this sum has a sufficient decay property, we obtain a universal bound on the Wasserstein distance between pinnings.

Formally, to track the growth of the Wasserstein distance at each step of the recursion, we define
\[
    \kappa_{s, \Delta} = \max_{G, i \in E} \W{\mu_{E-i; G}^{i\gets +},\mu_{E-i; G}^{i \gets -}},
\]
where the maximum is taken over all graphs $G$ with at most $s$ edges and maximum degree at most $\Delta$.
Intuitively, $\kappa_{s,\Delta}$ represents the worst-case Wasserstein distance when the graph has at most $s$ edges.
Using our expansion into weighted sums, we obtain a recursion of the form
\begin{align*}
    \kappa_{s,\Delta} \le 
    \begin{cases}
        0 & \text{if } s=1,\\
        \alpha (\kappa_{s-1,\Delta} + 1) & \text{otherwise},
    \end{cases}
\end{align*}
where $\alpha = \lambda \Delta / (1+\lambda \Delta)$.
Solving this shows that
\[
    \kappa_{s,\Delta} \le \frac{\alpha}{1-\alpha} = O(\lambda \Delta).
\]
Because we choose $\lambda = \Delta^{\Theta(1/\varepsilon)}$, we obtain the desired sensitivity bound.

\paragraph{Planar and Bipartite Graphs}
Recent work~\cite{chen2023near} showed that sensitivity can also be used to bound the spectral independence of the Gibbs distribution, which in turn yields a mixing-time bound for the Glauber dynamics.
An intriguing observation is that if we adopt the same recursion as in the bounded-degree setting, but measure the Wasserstein distance using the Hamming distance of matched \emph{vertices} (as opposed to edges), then the sensitivity can be bounded by $2$, independent of the parameter~$\lambda$ and the maximum degree!

Motivated by this observation, we propose a new form of Glauber dynamics that operates on \emph{matched vertices}.
In contrast to the standard edge-based Glauber dynamics, our approach employs a $2$-site update rule to maintain the connectivity of the Markov chain.
A similar Markov chain was studied in~\cite{alimohammadi2021fractionally} to count the number of size-$k$ matchings on planar graphs, and their method can be extended to approximately sample from the Gibbs distribution on planar graphs without dependence on $\lambda$.
But the time complexity of this sampling is high with respect to the size of the graph, as it requires a reduction from sampling to counting, followed by another reduction back to sampling.
We bypass these reductions, and our algorithm achieves significantly better time complexity with respect to the size of the graph, while maintaining only a logarithmic dependence on $\lambda$.

Concretely, we define a $2$-Glauber dynamics, denoted $\mathcal{P}_G$, on the distribution of matched vertices induced by the Gibbs distribution of matchings, which we denote by $\mu_{V; G}$. 
We then show that $\mathcal{P}_G$ is connected and converges to $\mu_{V; G}$.
By slightly modifying our earlier recursion for bounded-degree graphs, we derive an $O(1)$ bound on the sensitivity of $\mu_{V; G}$ and obtain a mixing time of
\[
   O\left(n^2\left(m\,\log \lambda +\log \frac{1}{\delta}\right)\right).
\]

Although a similar spectral independence bound has been established in \cite{alimohammadi2021fractionally}, this work presents the first sensitivity bound for $\mu_{V; G}$. Notably, our proof is grounded in the combinatorial structure of the distribution, differing fundamentally from the zero-freeness-based approach employed in \cite{alimohammadi2021fractionally}.

Furthermore, if there exists a fully polynomial-time (deterministic or randomized) approximation scheme (FPTAS/FPRAS) for counting perfect matchings in subgraphs of $G$, as is known for planar and bipartite graphs, we can implement $\mathcal{P}_G$ in polynomial time with sufficient accuracy to approximate $\mu_{V; G}$.
After obtaining a sample $U$ from $\mu_{V; G}$, we can then sample from $\mu_{E; G}$ by approximately sampling a perfect matching of $G[U]$, where $G[U]$ is a subgraph of $G$ induced by $U$, using a standard reduction that employs the approximate counting algorithm for perfect matchings on subgraphs of $G$.

\paragraph{General Graphs}
To obtain a low-sensitivity algorithm for general graphs, we consider sparsifying the input graph into a low-degree subgraph that approximately preserves the maximum matching size, while ensuring the sparsification process itself has low sensitivity. 
To this end, we leverage the recent notion of an \emph{$\varepsilon$-degree sparsifier}~\cite{chen2025entropy}.
Intuitively, a subgraph $H$ of $G = (V, E)$ is an $\varepsilon$-degree sparsifier if its maximum degree is small, say, $O(\log n)$, and the maximum matching size in $H$ approximates that in $G$ within a factor of $\varepsilon$.

Let $x \in \mathbb{R}^E$ be an optimal solution to the LP relaxation of the maximum matching problem with odd set constraints. 
That is, we maximize $\sum_{e \in E} x(e)$ subject to:
\[
\sum_{e \ni v} x(e) \leq 1 \quad \text{for all } v \in V, \quad \text{and} \quad \sum_{e \in E(G[B])} x(e) \leq \left\lfloor \frac{|B|}{2} \right\rfloor \quad \text{for all odd-sized } B \subseteq V.
\]
It was shown in~\cite{chen2025entropy} that sampling each edge $e \in E$ independently with probability $\Theta(x(e) /\gamma)$ yields an $\varepsilon$-degree sparsifier of $G$ with maximum degree $\Theta(\varepsilon^{-2} \log n)$ with high probability, where $\gamma= \Theta(\varepsilon^2/\log n)$.

However, this sparsification process does not inherently have low sensitivity: small changes in the input graph can lead to large changes in the LP solution $x$. 
To mitigate this, we introduce a entropy regularization term $\alpha \sum_{v \in V}H(\{x(e)\}_{e \in E: v \in E})$ to stabilize the LP, where $H(\{x(e)\}_{e \in F}) = \sum_{e \in F}x(e) \log (1/x(e))$ is the entropy function and $\alpha = \Theta(\varepsilon/\log n)$.
This does not significantly degrade the maximum matching size of the sparsifier $H$.
Let $x \in \mathbb{R}^E$ and $\tilde x \in \mathbb{R}^E$ denote the solutions to the regularized LPs for $G=(V,E)$ and $\tilde G = G-e:=(V,E \setminus \{e\})$ for some $e \in E$, respectively. 
We can then show that $\|x - \tilde x\|_1 = \tilde O(\sqrt{n/\alpha}) = \tilde O(\sqrt{\varepsilon^{-1} n})$.
As a result, the Wasserstein distance between the edge sets of the corresponding sparsifiers $H$ and $\tilde H$ is also $\tilde O(\sqrt{\varepsilon^{-1}n}/\gamma) = \varepsilon^{-O(1)} \sqrt{n}$. 
Finally, by applying our algorithm for low-degree graphs (\Cref{thm:bounded-degree-algorithm}) to the sparsifier, we obtain \Cref{thm:unbounded-degree-algorithm}.

\subsection{Related Work}\label{subsec:related}

\paragraph{Sensitivity}
In addition to the sensitivity upper bounds for maximum matching discussed above, it is known that any deterministic constant-factor approximation algorithm must incur a sensitivity of $\Omega(\log^* n)$, even on constant-degree graphs~\cite{yoshida2021sensitivity}, where $\log^* n$ denotes the iterated logarithm.
This result highlights the necessity of randomness in establishing \Cref{thm:bounded-degree-algorithm}, which achieves constant sensitivity for constant-degree graphs.

There are few problems for which algorithms with low sensitivity are known to exist. 
One example is the minimum $s$-$t$ cut problem, for which there is an algorithm that achieves an additive error of $O(n^{2/3})$ with sensitivity $O(n^{2/3})$~\cite{varma2023average}. 
Another example is a $(1 + \varepsilon)$-approximation algorithm for the shortest path problem with sensitivity $O(\varepsilon^{-1} \log^3 n)$, though in this case, sensitivity is measured with respect to edge contractions rather than edge deletions~\cite{kumabe2023lipschitz}. 
(Indeed, it is known that any algorithm that outputs an explicit $s$-$t$ path must have sensitivity $\Omega(n)$ with respect to edge deletions; see~\cite{varma2023average}.)

The idea of using the Gibbs distribution has previously been applied to design an algorithm for the minimum cut problem with low average sensitivity, where the average is taken over the choice of deleted edges~\cite{varma2023average}. 
However, the analysis of the approximation ratio, sensitivity, and running time in that setting relies heavily on Karger's cut-counting lemma, which states that the number of cuts of size at most $\alpha$ times the minimum cut is at most $O(n^{2\alpha})$~\cite{karger1993global}. 
In contrast, for the matching problem, such structure does not hold: there exist graphs with a unique maximum matching, yet an exponential number of matchings of size just one less (see \Cref{sec:lower-bound}). 
As a result, we take a completely different approach to analyzing the approximation ratio, sensitivity, and running time in our setting.

\paragraph{Sampling from the Gibbs Distribution over Matchings}
We discuss the problem of sampling from the Gibbs distribution over matchings, where $\lambda$ is the model parameter, $\delta$ is the desired total variation distance, and $\Delta$ denotes the maximum degree of the graph.  
In their classical work~\cite{jerrum1989approximating}, Jerrum and Sinclair established the rapid mixing of Glauber dynamics for this model. The best known mixing time until recently was  
\[
    O\left(nm \max\{1,\lambda\} \left(n(\log n + \log \max\{1,\lambda\}) + \log \frac{1}{\delta}\right)\right)
\]  
for arbitrary graphs~\cite[Proposition 12.4]{jerrum1996markov}.
Very recently, this was improved to $O(m \log n)$ on all bounded-degree graphs~\cite{chen2024fast,chen2024spectral}. 
Based on our own calculation, the precise mixing time is  
\[
\exp\left(O(\lambda^3\Delta^3 \log(\lambda^2\Delta^3))\right) \cdot m\left(\log n + \log \frac{1}{\delta}\right).
\]  
Our algorithm for bounded-degree graphs (\Cref{thm:bounded-degree-algorithm}) builds on this result.

For planar graphs, it has been shown that one can sample a matching in time $\mathrm{poly}(n, \log (1/\delta))$ even from the monomer-dimer model~\cite{alimohammadi2021fractionally}, where each matching $M$ is sampled with probability proportional to $\prod_{e \in M} w(e) \cdot \prod_{v \in V \setminus V(M)} \lambda(v)$ for weight functions $w:E \to \mathbb R_{\geq 0}$ and $\lambda:V \to \mathbb R_{\geq 0}$.

\paragraph{Differential Privacy}

Sensitivity is closely related to differential privacy~\cite{dwork2006differential}.  
In the context of graphs, for $\eta>0$, a randomized algorithm $A$ is said to be \emph{$(\eta,\delta)$-(edge) differentially private} if, for any graph $G, G'$ different by only one edge, and any set of outcomes $\mathcal{O}$ of the algorithm, the following inequality holds:
\[
    \Pr{A(G) \in \mathcal{O}} \leq \exp(\eta) \cdot \Pr{A(G) \in \mathcal{O}} + \delta.
\]

One way to design a differentially private algorithm for releasing statistics of a matching is to first compute a matching $M$ using \Cref{thm:bounded-degree-algorithm}, then compute a statistic of interest on $M$, and finally add Laplace noise calibrated to its sensitivity.  
To guarantee differential privacy in this way, we need the output of the algorithms to be stable not only in expectation but also with a high probability.
We can show that the algorithm presented in \Cref{thm:bounded-degree-algorithm} has such a guarantee (the proof is postponed to \Cref{sec:dp}):
\begin{lemma}\label{lem:DP}
    Let $0<\eps<1$ and let $G, G'$ be graphs different on only one edge with maximum degree $\Delta$.
    Denote by $\mathcal A_{\eps}$ the algorithm presented in \Cref{thm:bounded-degree-algorithm} with parameter $\eps$. 
    Then for any $c>0$, there exists a coupling between the output distributions $\mathcal A_\eps(G)$ and $\mathcal A_\eps(G')$ such that, with probability at least $1-n^{-c}$, the Hamming distance between their outputs is bounded by $c\Delta^{O(1/\eps)} \log n$.
\end{lemma}

For instance, this approach can be used to design an $(\eta,n^{-c})$-differentially private algorithm that releases the number of matched vertices with a given attribute with additive error of $c\eta^{-1}\Delta^{O(1/\varepsilon)}\log n$ for graphs of maximum degree $\Delta$.

We mention that it is known that outputting any subgraph, including a matching, is impossible under standard differential privacy. 
Consequently, previous works have explored relaxed notions, such as joint differential privacy~\cite{hsu2014private}, or the release of an implicit representation of a matching, from which each node can locally recover its matched edge using its own private information~\cite{dinitz2025differentially}.

\paragraph{Non-Signaling Model}

Although sensitivity is an intriguing concept on its own, it is tightly linked to the \emph{non-signaling} model \cite{akbari2025online,arfaoui2014can,coiteux2024no,gavoille2009can}.
In this framework, for a graph $G=(V,E)$, an algorithm may output any probability distribution, provided it respects the non-signaling principle: for every vertex set $S\subseteq V$, changes made to the input graph at distance $t+1$ or more from $S$ must leave the output distribution on $S$ unchanged.
The integer $t$ is called the model’s \emph{locality}. 
Because it imposes only this non-signaling constraint, the model is strictly more powerful than any ``physical'' synchronous distributed model with the same locality: it subsumes the classical $\mathsf{LOCAL}$ model \cite{linial1992locality}, the quantum-$\mathsf{LOCAL}$ model \cite{gavoille2009can}, and the bounded-dependence model \cite{holroyd2017finitary}, among others \cite{akbari2025online}.
It is recently observed in \cite{fleming2024sensitivity} that, if there is a distribution in the non-signaling model with locality $t$ on graphs with maximum degree at most $\Delta$, then we can use it to design an algorithm sensitivity $\Delta^{O(t)}$.

Recent work has shown that the KMW lower bounds \cite{coupette2021breezing,kuhn2016local}, originally established for the $\mathsf{LOCAL}$ model, also hold in the non-signaling model \cite{balliu2025new}.\footnote{They showed this fact for fractional problems, but it holds for constant-factor approximation to the maximum matching problem because the problem has an integrality gap of two, which is constant.}
Consequently, any non-signaling distribution that achieves a $(1/\mathrm{polylog} \Delta)$-approximation for the maximum matching problem must have locality $\Omega (\log \Delta / \log\log \Delta)$.
Even if such an optimal-locality distribution existed, the general correspondence between the non-signaling model and sensitivity \cite{fleming2024sensitivity} would still give only the sensitivity bound $\Delta^{O(\log \Delta / \log\log \Delta)}$ for $(1/\mathrm{polylog}\Delta)$-approximation.
By contrast, \Cref{thm:bounded-degree-algorithm} gives sensitivity $\Delta^{O(1/\varepsilon)}$ for $(1-\varepsilon)$-approximation, which implies that the Gibbs distribution is strictly more powerful than any non-signaling distribution for the maximum matching problem in terms of the trade-off between approximation ratio and sensitivity.

\subsection{Organization}
In \Cref{sec:prelim}, we introduce the key notions and basic properties used throughout the paper.
In \Cref{sec:sensitivity}, we present algorithms for bounded-degree graphs and prove \Cref{thm:bounded-degree-algorithm}.
In \Cref{sec:fast-algorithms}, we develop faster algorithms for planar and bipartite graphs, establishing \Cref{thm:planar-algorithm,thm:bipartite-algorithm}.
In \Cref{sec:sparsification}, we address general graphs and prove \Cref{thm:unbounded-degree-algorithm}.
Finally, in \Cref{sec:lower-bound}, we show that the parameter $\lambda$ in the Gibbs distribution must be at least $2^{\Omega(1/\varepsilon)}$ to achieve a $(1 - \varepsilon)$-approximation.

\section{Preliminaries}\label{sec:prelim}
All the proofs in this section are postponed to \Cref{sec:prelim-proof}.
\subsection{Graph Theory}
For a graph $G=(V,E)$, a matching is a set of edges in $E$ such that no two edges share a common vertex.
We denote the set of all matchings in $G$ by $\mathcal M(G)$.
Let $\nu(G)$ denote the maximum matching size of $G$.
For a vertex $v\in V$, we denote its incident edges in $G$ by $E_G(v)$.
We omit the subscript if it is clear from the context.
For a vertex set $S\subseteq V$, let $G[S]$ denote the subgraph induced by $S$.
For a set of edges $F\subseteq E$, let $V(F) \subseteq V$ denote the set of all vertices incident to at least one edge in $F$.
We denote by $G[F]$ the subgraph of $G$ consisting of $F$ and vertices incident to $F$, i.e., $G[F]=(V(F),F)$.
We denote by $G\setminus F$ the subgraph obtained by deleting $F$ from $G$.
When $F = \set e$, $G\setminus F$ is also denoted by $G - e$.
Let $\Delta(v)$ denote the degree of vertex $v \in V$.
We say an edge $e$ in a graph $G$ is \emph{pendant} if one of its endpoints is of degree $1$.

\subsection{Sensitivity}
In this section, we formally define the sensitivity of a graph algorithm.
\begin{definition}[Coupling between distributions]
    For two distributions $\mu, \mu'$ on the same probability space $\Omega$, we say a distribution $\pi$ on $\Omega\times\Omega$ is a \emph{coupling} between $\mu$ and $\mu'$ if the marginal distribution of $\pi$ on the first and second coordinates are $\mu$ and $\mu'$, respectively.
    We denote the set of all couplings between $\mu$ and $\mu'$ by $\Pi(\mu, \mu')$.
\end{definition}

\begin{definition}[Hamming distances on matchings]
    For a graph $G=(V, E)$, we define two metrics on $2^E$:
    \[
    d_V(F_1, F_2)\defeq |V(F_1)\mathbin \triangle V(F_2)|,
    \quad
    d_E(F_1, F_2)\defeq |F_1\mathbin \triangle F_2|,
    \]
    for any $F_1, F_2\in 2^E$,
    where $\triangle$ denotes the symmetric difference.
\end{definition}

\begin{definition}[Wasserstein distance]
    For two distributions $\mu, \mu'$ on the same probability space $\Omega$ with metric $d$, the ($1$-)Wasserstein distance is defined by
    \[
        \Wd{\mu, \mu'}\defeq \inf_{\pi\in\Pi(\mu, \mu')}\E[(x, y)\sim \pi]{d(x, y)}.
    \]
    When $\Omega = \mathcal M(G)$ for a graph $G=(V, E)$, we denote
    $\Wd[d_E]{\cdot,\cdot}$ and $\Wd[d_V]{\cdot,\cdot}$ by $\We{\cdot,\cdot}$ and $\Wv{\cdot,\cdot}$, respectively.
    We omit the superscript when the metric is clear from the context.
\end{definition}

\begin{definition}[Sensitivity]
    For a (randomized) algorithm $A$ that accepts a graph as an input
    and outputs an edge set, we define the sensitivity of $A$ on a graph $G$ by
    \[  
        \max_{e\in E(G)} \We{A(G), A(G-e)}.
    \]
\end{definition}

\subsection{Wasserstein Distance}
We regard a distribution $\mu$ on a finite set $\Omega$ as a function from $\Omega$ to $\mathbb R$ whose function value on $x$ is $\mu(x)$.
Hence, we can perform addition and scalar multiplication on distributions on a common space (though the resulting function may not represent a distribution).
The following lemma plays a vital role in this work, whose proof can be found in e.g., Theorem~1.14 in~\cite{villani2021topics}.
\begin{lemma}[Kantrovich-Rubinstein duality theorem]\label{lem:duality}
    For any two distributions $\pi, \pi'$ on a metric space $\Omega$, we have
    \[
    \W{\pi, \pi'} = \sup_{f\in L^1(\Omega)} \inner{f}{\pi-\pi'} = \sup_{f\in L^1(\Omega)}\sum_{x\in\Omega}f(x)(\pi(x)-\pi'(x)),
    \]
    where $L^1(\Omega)\defeq \set{f\in\mathbb R^{\Omega}\mid \forall x, y\in \Omega : |f(x)-f(y)| \le d(x, y)}$ denotes the set of $1$-Lipschitz functions. 
\end{lemma}
\begin{lemma}\label{lem:wd-sum}
    Let $\mu, \mu'$ be two distributions on the same finite metric space $\Omega$.
    Suppose that there exist two sets $\set{\mu_i}_{1\le i\le n}$ and $\set{\mu'_i}_{1\le i\le n}$ of distributions on $\Omega$ and a set $\set{\lambda_i}_{1\le i \le n}$
    of non-negative real numbers such that
    \[
    \mu - \mu' = \sum_{i=1}^n \lambda_i(\mu_i - \mu'_i).
    \]
    Then, we have
    \[
    \W{\mu, \mu'} \le \sum_{i=1}^n \lambda_i\W{\mu_i, \mu'_i}.
    \]
\end{lemma}
One immediate corollary is the triangle inequality:
\begin{corollary}[Triangle inequality]\label{lem:triangle}
    Let $\mu_1, \ldots,\mu_n$ be distributions on a common finite metric space $\Omega$,
    We have
    \[
    \W{\mu_1, \mu_n} \le \sum_{i=1}^{n-1} \W{\mu_i, \mu_{i+1}}.
    \]
\end{corollary}

\subsection{Gibbs Distributions}
We recall the definition of the Gibbs distribution for matchings, a special case of the monomer-dimer model, and derive its basic properties.

\begin{definition}[Gibbs distribution on matchings]
    Let $G=(V,E)$ be a graph.
    For a parameter $\lambda > 0$, the Gibbs distribution $\mu_{E;\lambda,G}$ over matchings in $G$ is defined as 
    \[
        \mu_{E; \lambda, G}(M) \defeq \frac1Z \lambda^{|M|}
    \]
    for every matching $M \in \mathcal M(G)$, where $Z\defeq\sum_{M\in\mathcal M(G)}\lambda^{|M|}$ is the partition function.
    When $\lambda$ is clear from the context, we sometimes omit $\lambda$ and 
    simply write $\mu_{E;G}$.
\end{definition}

\begin{definition}[Marginal and conditional distribution]\label{def:marginal-conditional}
    The marginal distribution of $\mu_{E;\lambda,G}$ on a subset $F \subseteq E$, denoted $\mu_{F;\lambda, G}$, is defined as a distribution over $\mathcal M(G[F])$ such that
    \[
    \mu_{F;\lambda, G}(N) \defeq \frac1Z
    \sum_{\substack{M\in\mathcal M(G)\\M\cap F= N}}
    \lambda^{|M|}
    \]
    for every matching $N \in \mathcal M(G[F])$, where $Z\defeq\sum_{M\in\mathcal M(G)}\lambda^{|M|}$ is the same partition function as $\mu_{E;\lambda,G}$.
    
    We denote the events that an $e\in E$ is matched and is not matched by $e\gets +$ and $e\gets -$, respectively, which are  referred to as \emph{pinnings}.
    With a bit abuse of notation, for a subset $F\in E$, we denote by $F\gets -$ (resp., $F \not  \gets -)$ the event that none (resp., some) of the edges in $F$ is matched.
    For a vertex $v \in V$, we denote by $v \gets +$ (resp., $v \gets -)$ the event that some (resp., none) of the edges incident to $v$ is matched. 
    
    We write the event that a conditional distribution is conditioned on as
    a superscript. For example $\mu_{E;\lambda,G}^{e\gets +}$ is the distribution
    conditioned on $e$ is matched. When we condition on multiple pinnings
    at the same time, we simply stack the superscripts up.
\end{definition}
One basic property of the Gibbs distribution is the following lemma that asserts the equivalence between graph modification and pinning.
For an edge $i \in E$ (we often use $i$ and $j$ to denote edges when analyzing properties of the Gibbs distribution), let $N(i)$ denote the set of edges that are incident to $i$, excluding $i$ itself.
\renewcommand{\mp}[3]{\ifblank{#3}{\mu_{#1; \lambda, #2}^{\tau }}{\mu_{#1; \lambda, #2}^{\tau, #3} }}
\newcommand{\mt}[2]{\ifblank{#2}{ \mathcal M^{\tau} (#1) }{ \mathcal M^{\tau, #2} (#1) }}
\begin{lemma}\label{lem:pinning-deletion}
    Let $G=(V, E)$ be a graph, $i\in E$ be an edge, $\lambda > 0$,
    $E' \subseteq E$ and $F\subseteq N(i)$. Then
    \[
        \mp {E'} G {i\gets -} = 
        \mp {E'-i} {G-i} {},
        \quad
        \mp {E - i} G {i\gets +} = 
        \mp {E\setminus F - i} {G\setminus F} {i\gets +},
        \quad
        \mp {E - i} G {i\gets +} = 
        \mp {E\setminus N(i) - i} {G\setminus N(i)-i} {}.
    \]
\end{lemma}

Since for a vertex $v \in V$ there is at most one edge matched in $E(v)$,
we can decompose the Gibbs distribution according to the 
matched edges in $E(v)$, resulting in the following lemma.
\begin{lemma}\label{lem:expansion}
    Let $G = (V, E)$ be a graph, $v\in V$ be a vertex, and $\lambda>0$ be the parameter of the Gibbs distributions.
    Then, we have
    \[
        \mu_{E;G} =
        \sum_{i\in E(v)}\mu_{E;G}(i\gets+)\mu_{E; G}^{i\gets +} + \mu_{E;G}(E(v)\gets-)\mu_{E;G}^{E(v)\gets -}.
    \]
\end{lemma}

\subsection{Glauber dynamics}
To sample from the Gibbs distribution, a natural approach is to use a Markov chain whose stationary distribution coincides with the target distribution. 
In this context, one of the most widely studied dynamics is the Glauber dynamics, a local Markov chain that updates the configuration by modifying a small part at each step:
\begin{definition}\label{def:glauber-dynamics}
    Let $\mu$ be a distribution over $2^\Omega$ for some finite set $\Omega$ and $1\le k\le |\Omega|$, we define the $k$-site Glauber dynamics $\mathcal P$ of $\mu$ as a Markov chain defined on the support of $\mu$ as follows.
    Let $\binom{\Omega}{k}$ denote the family of subsets of size $k$.
    For $A\in \binom \Omega k$ and $x\in 2^\Omega$, define $B_A(x) \defeq \set{y\in 2^\Omega\mid x\mathbin \triangle y\subseteq A}$.
    Then for any $x, y$ with non-zero probability we define
    \begin{align*}
        \mathcal P(x\to y)\defeq
           \binom{|\Omega|}{k}^{-1} 
           \sum_{A\in \binom\Omega k : y\in B_A(x)}
           \frac{\mu(y)}{\sum_{z\in B_A(x)} \mu(z)}.
    \end{align*}
\end{definition}
Intuitively, each step of $P$ with the current state $x$ does the following:
\begin{enumerate}
    \item uniformly sample a size-$k$ subset $A$ of $\Omega$;
    \item transit to the next state $y$ with probability proportional to $\mu(y)$ under the condition that $x\mathbin \triangle y\subseteq A$.
\end{enumerate}
Notice that if $\mathcal P$ is the $k$-site Glauber dynamics of $\mu$, then $\mu(x) \mathcal P(x\to y) = \mu(y)\mathcal P(y\to x)$ for all $x, y\in 2^\Omega$, and hence if $\mathcal P$ is irreducible, then $\mu$ is its unique stationary distribution, and we have $\lim_{t\to\infty}\mu_0\mathcal P^t = \mu$ for any initial distribution $\mu_0$ on the support of $\mu$.
What we care about is the speed of convergence.
Letting $\delta > 0$, we define the \emph{mixing time} of $\mathcal P$ as
\begin{align*}
    t_{\mathrm{mix}}(\mathcal P, \delta) \defeq
    \min
    \set{
    t\in \mathbb Z_{\ge 0}
    \mid
    \forall \mu_0 : d_{\mathrm{tv}}(\mu_0 \mathcal P^t, \mu)< \delta
    }.
\end{align*}

\section{Low-Sensitivity Algorithms for Bounded-Degree Graphs}\label{sec:sensitivity}
\newcommand{\m}{\mu_{E;\lambda, G}}
\renewcommand{\mp}{ \mu_{E-i; \lambda, G}^{\substack{            i\gets +}}}
\newcommand{\mm}{ \mu_{E-i; \lambda, G}^{\substack{            i\gets -}}}
\newcommand{\mmp}{\mu_{E-i; \lambda, G}^{\substack{j\gets + \\ i\gets -}}}
\newcommand{\mmm}{\mu_{E-i; \lambda, G}^{\substack{v\gets - \\ i\gets -}}}
\newcommand{\pmp}{\mu_{E-i; \lambda, G}^{\substack{            i\gets -}}(j\gets +)}
\newcommand{\pmm}{\mu_{E-i; \lambda, G}^{\substack{            i\gets -}}(v\gets -)}
\newcommand{\se}{\sum_{j\in E(v)}}
\newcommand{\ke}[1]{\kappa^E_{#1,\lambda,\Delta}}
\newcommand{\megp}[3]{\mu_{#1 ; \lambda, #2}^{ #3 }}
In this section, we prove \Cref{thm:bounded-degree-algorithm}.

We consider an (inefficient) algorithm that, given a graph $G=(V,E)$, simply outputs a matching $M \subseteq E$ sampled from $\mu_{E;\lambda,G}$ for $\lambda=\Delta^{\Theta(1/\varepsilon)}$.
We show that this algorithm has approximation ratio $1-\varepsilon$ in Section~\ref{subsec:matching-approximation}, and we analyze its sensitivity  in Section~\ref{subsec:matching-sensitivity}.
Then in Section~\ref{subsec:matching-polynomial-time}, we discuss a modification to the algorithm that enables it to run in polynomial time.

\subsection{Approximation Guarantee}\label{subsec:matching-approximation}
In this section, we show the following:
\begin{theorem}\label{thm:matching-approximation}
    We have 
    \[
        \E[M \sim \mu_\lambda]{|M|} \geq  \frac{(1-\varepsilon/2)\lambda}{\lambda + (4\Delta)^{2\varepsilon^{-1}}} \cdot \nu(G).
    \]
    In particular, for $\lambda \geq 2\varepsilon^{-1}(4\Delta)^{2\varepsilon^{-1}}$, we have
    \[
        \E[M \sim \mu_\lambda]{|M|} \geq \frac{1-\varepsilon/2}{1+\varepsilon/2} \cdot \nu(G) \geq (1 - \varepsilon) \nu(G).
    \]
\end{theorem}

Let $G=(V,E)$ be a graph on $n$ vertices with maximum degree $\Delta$.
Let $m_k$ be the number of matchings of size $k$ in $G$.
The \emph{matching polynomial} $M_G$ and \emph{matching generating polynomial} $m_G$ are defined as 
\begin{align*}
    M_G(x) :=\sum_{k\geq 0}(-1)^{k}m_k x^{n-2k} \quad  \text{and} \quad
    m_G(x) :=\sum_{k\geq 0}m_k x^k,
\end{align*}
respectively.
Note that $M_{G}(x)=x^n m_{G}(-x^{-2})$ holds and that $m_G(\lambda)$ coincides with the partition function of the Gibbs distribution.
We start with the following simple fact:
\begin{lemma}\label{lem:matching-size-root}
    For the maximum matching size $\nu := \nu(G)$, let $\lambda_1,\ldots,\lambda_\nu$ be the roots of $m_G$.
    Then, we have 
    \[
        \E[M \sim \mu_\lambda]{|M|} = \sum_{i=1}^\nu \frac{\lambda}{\lambda - \lambda_i}.
    \]
\end{lemma}
\begin{proof}
    Let $m'_G(\lambda)$ be the derivative of $m_G$ at $\lambda$.
    Then, we have
    \[
        \E[M \sim \mu_\lambda]{|M|}
    = \frac{\sum_{k \geq 0}k m_k \lambda^k}{\sum_{k \geq 0}m_k \lambda^k}  
    = \frac{\lambda m'_G(\lambda)}{m_G(\lambda)} 
    = \frac{\lambda m_\nu \sum_{i=1}^\nu \prod_{j\neq i}(\lambda-\lambda_i) }{m_\nu \prod_{i=1}^\nu (\lambda-\lambda_i)}
    =
    \sum_{i=1}^\nu \frac{\lambda}{\lambda-\lambda_i}.
    \]
\end{proof}
Next, we show that the absolute values of the roots of the matching generating polynomial are bounded.
To this end, we use the following property of the matching polynomial.
\begin{lemma}[\cite{godsil1981matchings,heilmann1972theory}]
    The matching polynomial is real-rooted.
    Moreover, all the roots have absolute values at most $2\sqrt{\Delta-1}$.
\end{lemma}

\begin{corollary}\label{cor:lower-bounds-of-roots-of-MGP}
    All the roots of generating matching polynomial are real and negative.
    Moreover, all the roots have absolute value at least $1/(4(\Delta-1))$.
\end{corollary}
\begin{proof}
    Obvious from the formula $M_{G}(x)=x^n m_{G}(-x^{-2})$.
\end{proof}

\begin{lemma}\label{lem:upper-bounds-of-roots-of-MGP}
    For any $\varepsilon > 0$, at least a $(1-\varepsilon)$-fraction of the roots of the matching generating polynomial has absolute value at most $(4\Delta)^{\varepsilon^{-1}}$.
\end{lemma}
\begin{proof}
    Let $\lambda_1,\ldots,\lambda_\nu$ be the roots of $m_G$ for   $\nu := \nu(G)$.
    Then, we have $\prod_{i=1}^\nu \lambda_i = (-1)^\nu m_0/m_\nu = (-1)^\nu/m_\nu$ by Vieta's formula.
    As all the roots are negative by \Cref{cor:lower-bounds-of-roots-of-MGP}, this implies $\sum_{i=1}^\nu \log |\lambda_i| = -\log m_\nu$.
    By \Cref{cor:lower-bounds-of-roots-of-MGP}, we have $\log |\lambda_i| \geq -\log (4(\Delta-1)) \geq -\log (4\Delta)$ for every $i \in [\nu]$.

    Suppose that there are $\varepsilon \nu$ many roots $\lambda_i$ with $\log |\lambda_i| > \tau$ for some $\tau \in \mathbb{R}$.
    Then, we have
    \[
        -\log m_\nu = \sum_{i=1}^\nu \log |\lambda_i| \geq \varepsilon \nu \tau - (1-\varepsilon) \nu \log (4\Delta).
    \]
    Hence, we have 
    \[
        \tau \leq \frac{(1-\varepsilon) \log (4\Delta)}{\varepsilon} - \frac{\log m_\nu }{\varepsilon \nu }
        \leq \frac{1-\varepsilon}{\varepsilon} \log (4\Delta).
    \]
    This implies that at least a $(1-\varepsilon)$-fraction of roots have absolute value at most $(4\Delta)^{(1-\varepsilon)/\varepsilon} \leq (4\Delta)^{\varepsilon^{-1}}$.
\end{proof}

\begin{proof}[Proof of \Cref{thm:matching-approximation}]
    Let $\lambda_1,\ldots,\lambda_\nu$ be the roots of $m_G$ for $\nu := \nu(G)$.
    Then, the expected size of a matching sampled from the Gibbs distribution is 
    \begin{align*}
        & \E[M \sim \mu_\lambda]{|M|}
        = \sum_{i=1}^\nu \frac{\lambda}{\lambda-\lambda_i} \tag{by \Cref{lem:matching-size-root}} \\
        & \geq 
        \sum_{i=1}^{(1-\varepsilon/2)\nu} \frac{\lambda}{\lambda + (4\Delta)^{2\varepsilon^{-1}}} \tag{by \Cref{lem:upper-bounds-of-roots-of-MGP}} \\
        & =  \frac{(1-\varepsilon/2)\lambda}{\lambda + (4\Delta)^{2\varepsilon^{-1}}}\cdot \nu. 
    \end{align*}
\end{proof}

\subsection{Sensitivity Analysis}\label{subsec:matching-sensitivity}

In this section, we prove the following theorem about sensitivity of the Gibbs distribution over matchings:
\begin{theorem}\label{thm:sensitivity}
    Let $G = (V, E)$ be a graph of maximum degree $\Delta$, $i\in E$ be an edge, and $\lambda > 0$. 
    Then, we have
    \[
        \We{\mu_{E; \lambda, G}, \mu_{E; \lambda, G-i}} \le 1 + 2\lambda\Delta.
    \]
\end{theorem}
We first consider sensitivity with respect to pinning in \Cref{susubsec:pinning}, using an auxiliary lemma that will be proved in \Cref{subsubsec:recursion}. Then, we prove \Cref{thm:sensitivity} in \Cref{subsubsec:deletion}.

\subsubsection{Sensitivity with Respect to Pinning}\label{susubsec:pinning}
In this section, we consider the Wasserstein distance with respect to pinning, and show the following:
\begin{lemma}\label{lem:pinning}
    Let $G = (V, E)$ be a graph of maximum degree $\Delta$, $i\in E$ be an edge, and $\lambda > 0$. 
    Then, we have
    \begin{equation*}
        \We{\mu_{E; \lambda, G}^{i\gets +}, \mu_{E; \lambda, G}^{i\gets -}} \le  1 + 2\lambda\Delta.
    \end{equation*}
\end{lemma}

We prove \Cref{lem:pinning} by a recursive coupling similar to that used in \cite{chen2023near,chen2024fast}, but 
instead of using an algorithm that outputs the coupling then analyzing the output of the algorithm as they did, we construct the recursion of Wasserstein distance
directly with the help of \Cref{lem:wd-sum}.

For each $s\ge 1$, we define an upper bound on the Wasserstein distance on graphs with at most $s$ edges, and put them into a recursion to derive a universal upper bound for graphs of any size.
\begin{definition}[$\kappa^E_{s, \Delta, \lambda}$]\label{def:kappa}
    For each $s, \Delta\in \mathbb Z_{>0}$, $\lambda>0$, we define
    \begin{align*}
    \kappa^E_{s, \Delta, \lambda}
    &\defeq
    \max_{\substack{G=(E, V)\\i\in E\\i\text{ is pendant}}} \We{\mu_{E-i;\lambda, G}^{i\gets +},\mu_{E-i;\lambda, G}^{i\gets -}},
    \end{align*}
    where $G=(V,E)$ ranges over all graphs such that $|E|\le s$ and $\Delta(G)\le \Delta$.
\end{definition}

The following lemma shows the recurrence relation among the upper bounds.
\begin{lemma}\label{lem:kappa-recursion}
    For any $s, \Delta\in \mathbb Z_{>0}$, $\lambda>0$, we have
    \begin{align*}
    \kappa^E_{s+1, \Delta, \lambda}
    &\le \frac{\lambda\Delta}{1+\lambda\Delta}
    \left(\kappa^E_{s, \Delta, \lambda}+1\right).
    \end{align*}
\end{lemma}
We will prove \Cref{lem:kappa-recursion} in \Cref{subsubsec:recursion}; here, we first use it to prove \Cref{lem:pinning}.
\begin{proof}[Proof of \Cref{lem:pinning}]
    Notice that $\kappa^E_{1,\Delta,\lambda}=0$ holds trivially.
    By \Cref{lem:kappa-recursion}, we can prove by induction that for any $s\ge 1$,
    \begin{align*}
    \kappa^E_{s, \Delta, \lambda}
    &\le \sum_{k=1}^{s-1}\left(\frac{\lambda\Delta}{1+\lambda\Delta}\right)^k
    \le \sum_{k=1}^{\infty}\left(\frac{\lambda\Delta}{1+\lambda\Delta}\right)^k
    = \lambda\Delta.
    \end{align*}
    By \Cref{def:kappa}, this indicates that for any graph $G=(V, E)$ such that $\Delta(G)\le \Delta$,
    and any pendant edge $i\in E$, we have
    \begin{align}
      &\We{\mu_{E-i;\lambda, G}^{i\gets +},\mu_{E-i;\lambda, G}^{i\gets -}} \le \lambda\Delta.
      \label{eq:unibound-kappa-e}
    \end{align}
    
    To show the claim, it suffices to reduce deleting an arbitrary edge $e \in E$ from a graph $G=(V, E)$ to the case where pendant edges are deleted.
    To this end, we define a new graph $G' = (V, E')$ by breaking the edge $i$ into two edges $i_1, i_2$,
    each incident to one of its endpoints. Then we have 
    \[
    \mu_{E-i; G, \lambda}^{i\gets +}
    = \mu_{E'\setminus\set{i_1, i_2}; G', \lambda}^{\substack{i_1\gets + \\ i_2\gets +}},
    \qquad
    \mu_{E-i; G, \lambda}^{i\gets -}
    = \mu_{E'\setminus\set{i_1, i_2}; G', \lambda}^{\substack{i_1\gets - \\ i_2\gets -}}.
    \]
    Notice that $i_1, i_2$ are pendant edges in $G'$.
    By triangle inequality and \Cref{eq:unibound-kappa-e}, we have
   \begin{align*}
   \We{\mu_{E-i;\lambda, G}^{i\gets +},\mu_{E-i;\lambda, G}^{i\gets -}}
   &=
   \We{
   \mu_{E'\setminus\set{i_1, i_2}; G', \lambda}^{\substack{i_1\gets + \\ i_2\gets +}},
   \mu_{E'\setminus\set{i_1, i_2}; G', \lambda}^{\substack{i_1\gets - \\ i_2\gets -}}
   }
   \\&\le
   \We{
   \mu_{E'\setminus\set{i_1, i_2}; G', \lambda}^{\substack{i_1\gets + \\ i_2\gets +}},
   \mu_{E'\setminus\set{i_1, i_2}; G', \lambda}^{\substack{i_1\gets - \\ i_2\gets +}}
   }
   +\We{
   \mu_{E'\setminus\set{i_1, i_2}; G', \lambda}^{\substack{i_1\gets - \\ i_2\gets +}},
   \mu_{E'\setminus\set{i_1, i_2}; G', \lambda}^{\substack{i_1\gets - \\ i_2\gets -}}
   }
   \\&\le
   2\lambda\Delta.
   \end{align*}
   Moreover, we have 
   \begin{align*}
       \We{\megp {E-i} G {i\gets +}, \megp {E} G {i\gets +} } = 1,\quad
       \We{\megp {E-i} G {i\gets -}, \megp {E} G {i\gets -} } = 0.
   \end{align*}
   Then by triangle inequality, 
   \begin{align*}
   \We{\mu_{E;\lambda, G}^{i\gets +},\mu_{E;\lambda, G}^{i\gets -}}
   &\le
   \We{\megp {E} G {i\gets +}, \megp {E-i} G {i\gets +} }
   +
   \We{\mu_{E-i;\lambda, G}^{i\gets +},\mu_{E-i;\lambda, G}^{i\gets -}}
   +
   \We{\megp {E-i} G {i\gets -}, \megp {E} G {i\gets -} }
   \\&\le 1 + 2\lambda\Delta.
   \end{align*}
   
\end{proof}

\subsubsection{Recursive Coupling under Edge Pinnings}\label{subsubsec:recursion}
In this section, we prove \Cref{lem:kappa-recursion}.
The following lemma states that the probability that a vertex is matched is bounded away from one.
\begin{lemma}\label{lem:marginal-bound}
    Let $G=(V, E)$ be a graph, $v\in V$ be a vertex and $\lambda>0$.
    Then $\mu_{E; \lambda, G}(v\gets +)\le \frac{\lambda\Delta(v)}{1+\lambda\Delta(v)}$.
\end{lemma}
\begin{proof}
    In this proof, we write $\mathcal M$ instead of $\mathcal M(G)$ for notational convenience.
    For an edge $i\in E$, define 
    \begin{align*}
        \mathcal M^{i \gets +} & := \set{M\in \mathcal M\mid i\in M}, \\
        \mathcal M^{N(i) \gets -} & := \set{M\in\mathcal M \mid i\notin M, N(i)\cap M=\varnothing}.
    \end{align*}
    Then, there is a bijection between $\mathcal M^{i \gets +}$ and $\mathcal M^{N(i) \gets -}$.
    \begin{align*}
        \mathcal M^{i \gets +} \ni M &\leftrightarrow M - i \in \mathcal M^{N(i) \gets -}.
    \end{align*}
    Notice that if $v\gets -$ then we must have $E(v)\gets -$, so $\m(v\gets -)\ge \m(N(i)\gets -)$.
    Hence, we have
    \begin{align*}
        \frac{\m(i\gets +)}{\m(v\gets -)}
        &\le \frac{\m(i\gets +)}{\m(N(i)\gets -)}
        \\&\le
        \frac
        {\sum_{M \in \mathcal M^{i \gets +}}\m(M)}
        {
        \sum_{M\in \mathcal M^{N(i)\gets -}}\m(M)
        }
        \\&=
        \frac
        {\sum_{M \in \mathcal M^{i \gets +}}\m(M)}
        {
        \sum_{M\in \mathcal M^{i\gets +}}\m(M)/\lambda
        }= \lambda.
    \end{align*}
    Then, we have
    \begin{align*}
        \m(v\gets + )
        &=
        \frac{\sum_{i\in E(v)}\m(i\gets +)}{\m(v\gets -) + \sum_{i\in E(v)}\m(i\gets +)}
        \\&=
        \frac{\sum_{i\in E(v)}\m(i\gets +)/\m(v\gets -) }{1 + \sum_{i\in E(v)}\m(i\gets +)/\m(v\gets -) }
        \\&\le
        \frac{\lambda\Delta(v)}{1+\lambda\Delta(v)}.
    \end{align*}
\end{proof}
The following lemma is useful for our recursive argument.
\begin{lemma}\label{lem:to-kappa-edge}
    Let $G=(V, E)$ be a graph such that $\Delta(G)\le \Delta$, $|E|\le s$, $i = (u, v) \in E$ be a pendant edge such that $\Delta(u)=1$
    and $j\in E(v)$ be an edge other than $i$.
    Then
    \begin{enumerate}
        \item 
        $\We{\mu_{E-i; \lambda, G}^{i\gets +}, \mu_{E-i; \lambda, G}^{\substack{v\gets -\\i\gets -}}} = 0$;
        \item 
        $\We{\mu_{E-i; \lambda, G}^{i\gets +}, \mu_{E-i; \lambda, G}^{\substack{j\gets +\\i\gets -}}} \le \kappa_{s-1; \Delta, \lambda} + 1$.
    \end{enumerate}
\end{lemma}
\begin{proof}
    By \Cref{lem:pinning-deletion}, we have
    \[
    \We{\mu_{E-i; \lambda, G}^{i\gets +}, \mu_{E-i; \lambda, G}^{\substack{v\gets -\\i\gets +}}}
    =
    \We{\mu_{E\setminus E(v); \lambda, G\setminus E(V)},
        \mu_{E\setminus E(v); \lambda, G\setminus E(v)}}
    =0,
    \]
    proving the first part of the statement.

    For the second part, first note that
    \newcommand{\mupm}{\mu_{E-i; \lambda, G}^{\substack{j\gets -\\i\gets +}}}
    \newcommand{\mump}{\mu_{E-i; \lambda, G}^{\substack{j\gets +\\i\gets -}}}
    \[
    \We{\mu_{E-i; \lambda, G}^{i\gets +}, \mu_{E-i-j; \lambda, G}^{i\gets +}} = 0,
    \quad
    \We{\mu_{E-i-j; \lambda, G}^{i\gets +}, \mu_{E-i-j; \lambda, G}^{\substack{j\gets - \\ i\gets +}}} = 0,
    \quad
    \We{\mu_{E-i-j; \lambda, G}^{\substack{j\gets +\\i\gets -}}, \mu_{E-i; \lambda, G}^{\substack{j\gets +\\i\gets -}}} = 1.
    \]
    An application of triangle inequality gives
    \begin{align*}
    \We{\mu_{E-i; \lambda, G}^{i\gets +},
        \mu_{E-i; \lambda, G}^{\substack{j\gets + \\ i\gets -}}}
    &\le
    \We{\mu_{E-i  ; \lambda, G}^{i\gets +},
        \mu_{E-i-j; \lambda, G}^{i\gets +}}
    +
    \We{\mu_{E-i-j; \lambda, G}^{                    i\gets + },
        \mu_{E-i-j; \lambda, G}^{\substack{j\gets -\\i\gets +}}}
    \\&\phantom{=}+
    \We{\mu_{E-i-j; \lambda, G}^{\substack{j\gets -\\i\gets +}},
        \mu_{E-i-j; \lambda, G}^{\substack{j\gets +\\i\gets -}}}
    +
    \We{\mu_{E-i-j; \lambda, G}^{\substack{j\gets +\\i\gets -}},
        \mu_{E-i  ; \lambda, G}^{\substack{j\gets +\\i\gets -}}}
    \\&=
    0+0+
    \We{\mu_{E-i-j; \lambda, G}^{\substack{j\gets -\\i\gets +}},
        \mu_{E-i-j; \lambda, G}^{\substack{j\gets +\\i\gets -}}}
    +1.
    \end{align*}
    Notice that all edges in $E(v) - j$ have no chance to be matched in both $\mupm$ and $\mump$, so we may
    remove $E(i)-j$ from the scope to get
    \[
    \We{\mu_{E-i; \lambda, G}^{i\gets +},
        \mu_{E-i; \lambda, G}^{\substack{j\gets + \\ i\gets -}}}
    \le
    \We{\mu_{E-(E(v)-j); \lambda, G-(E(v)-j)}^{                    i\gets + },
        \mu_{E-(E(v)-j); \lambda, G-(E(v)-j)}^{                    i\gets - }}
    +1.
    \]
    Notice that $j$ is pendant in $E-(E(v)-j)$, and hence this distance is bounded from above by $\kappa_{s-1,\lambda,\Delta}$, proving the second part of the statement.
\end{proof}

With \Cref{lem:marginal-bound,lem:to-kappa-edge}, we can prove \Cref{lem:kappa-recursion}.
\begin{proof}[Proof of \Cref{lem:kappa-recursion}]
    Let $G=(V, E)$ be a graph and $i=(u, v)\in E$ be an pendant edge that fits the definition of $\kappa^E_{s,\Delta,\lambda}$.
    We may assume that $\Delta(u)=1$, then by \Cref{lem:expansion} we have
    \begin{align*}
        \mp - \mm
        &=
        \mp - \se\pmp\mmp -\pmm\mmm
        \\&=
        \se \pmp \left( \mp - \mmp \right) \\&\phantom{=}+ \pmm \left( \mp - \mmm \right).
    \end{align*}
    By \Cref{lem:wd-sum}, we have
    \begin{align*}
        \We{\mp , \mm}
        &\le
        \se \pmp \We{ \mp , \mmp} \\&\phantom{=}+ \pmm \We{\mp - \mmm}.
    \end{align*}
    By \Cref{lem:to-kappa-edge,lem:marginal-bound}, we have
    \begin{align*}
        \We{\mp , \mm}
        &\le
        \se \pmp\left(\ke{s-1} + 1\right)
        \\&=\m(v\gets +)\left(\ke{s-1} + 1\right)
        \\&\le \frac{\lambda\Delta}{1+\lambda\Delta}\left(\ke{s-1} + 1\right).
    \end{align*}
\end{proof}

\subsubsection{Sensitivity with Respect to Edge Deletion}\label{subsubsec:deletion}
\begin{proof}[Proof of \Cref{thm:sensitivity}]
    We have
    \begin{align*}
    \megp{E}{G}{} - \megp{E}{G}{i\gets -}
    &=
    \megp E G {} (i\gets +)
    \megp E G    {i\gets +}
    +
    \megp E G {} (i\gets -)
    \left(
    \megp E G    {i\gets -}
    -
    \megp{E}{G}{i\gets -}
    \right)
    \\&=
    \megp E G {} (i\gets +)
    \left(
    \megp E G    {i\gets +}
    -
    \megp E G    {i\gets -}
    \right).
    \end{align*}
    Then by \Cref{lem:wd-sum,lem:pinning}, we have
    \begin{align*}
        \We{\mu_{E; \lambda, G}, \mu_{E; \lambda, G-i}}
        & = \We{
        \megp E G {},
        \megp E G {i\gets -}
        } \\
        & \le
        \megp E G {} (i\gets +)
        \We{
        \megp E G    {i\gets +}
        ,
        \megp E G    {i\gets -}
        }
        \\&\le
        \We{
        \megp E G    {i\gets +}
        ,
        \megp E G    {i\gets -}
        }
        \\&\le 1 + 2\lambda\Delta.
    \end{align*}
\end{proof}

\subsection{Polynomial-Time Algorithm}\label{subsec:matching-polynomial-time}
There are two different bounds for the mixing time of the Glauber dynamics for $\mu$.
The first one is near-linear with the number of edges, while the second one is near-linear on the parameter $\lambda$.
\begin{lemma}[{\cite[Theorem 1.1]{chen2024fast}}]\label{lem:glauber-edge-time}
    Denote the (1-site) Glauber dynamics of $\m$ on edges by $\mathcal P$, then we have
    \begin{align*}
        t_{\mathrm{mix}}(\mathcal P, \delta) \le 
        \exp\left(O(\lambda^3\Delta^3 \log(\lambda^2\Delta^3))\right) m\left(\log n + \log \frac{1}{\delta}\right).
    \end{align*}
\end{lemma}

\begin{lemma}[{\cite[Proposition 12.4]{jerrum1989approximating}}]\label{lem:glauber-edge-time-jerrum}
    Denote the (1-site) Glauber dynamics of $\mu_{E;\lambda,G}$ on edges by $\mathcal P$, then we have
    \[
        t_{\mathrm{mix}}(\mathcal P, \delta) =
        O\left(nm \max\{1,\lambda\} \left(n(\log n + \log \max\{1,\lambda\}) + \log \frac{1}{\delta}\right)\right).
    \]
\end{lemma}

\begin{proof}[Proof of \Cref{thm:bounded-degree-algorithm}]
    Consider an algorithm that samples a matching by using the $1$-site Glauber dynamics $\mathcal P$ in \Cref{lem:glauber-edge-time} with the error parameter $\delta =1/n^2$ and $\lambda = 2(\varepsilon/2)^{-1}(4\Delta)^{1/\varepsilon}$, and outputs the sampled matching.
    By \Cref{thm:matching-approximation}, its approximation ratio is at least
    \[
        \left(1 - \frac{1}{n^2}\right)\left(1-\frac{\varepsilon}{2}\right) > 1-\varepsilon.
    \]
    By \Cref{thm:sensitivity}, the sensitivity is
    \[
        \left(1 - \frac{1}{n^2}\right) \cdot O(\lambda\Delta ) + \frac{1}{n^2} \cdot n
        = \Delta^{O(1/\varepsilon)}.
    \]
    By \Cref{def:glauber-dynamics}, the $1$-site Glauber dynamics of $\m$ is a Markov chain that does the transition as the following with current state $x\in\mathcal M(G)$:
    \begin{enumerate}
        \item choose $i\in E$ uniformly at random;
        \item transit to $x\mathbin \triangle \{i\}$ with probability $\frac{\m(x\mathbin\triangle \{i\})}{\m(x\mathbin\triangle \{i\}) +\m(x)}$, and stay at $x$ otherwise.
    \end{enumerate}
    Notice that
    \begin{align*}
        \frac{\m(x\mathbin \triangle \{i\})}{\m(x\mathbin \triangle \{i\}) +\m(x)}
        =
        \begin{cases}
           \frac{1}{1+\lambda}
           &
           i\in x
           \\
           \frac{\lambda}{1+\lambda}
           &
           i\notin x, x\cup \{i\}\in\mathcal M(G)
           \\
           0
           &
           i\notin x, x\cup \{i\}\notin\mathcal M(G)
        \end{cases}.
    \end{align*}
    By maintaining a list of matched vertices, we can check whether $x\cup \{i\}\in\mathcal M(G)$ in constant time,
    and hence each step of the Glauber dynamics can be run in $O(1)$ time.
    Moreover, we need $O(n+m)$ time to read the input graph.
    By \Cref{lem:glauber-edge-time},
    the time complexity is 
    \begin{align*}
        \exp(\Delta^{O(1/\varepsilon)})m\log n + O(m+n) = \exp(\Delta^{O(1/\varepsilon)})m\log n + O(n).
    \end{align*}
    On the other hand, \Cref{lem:glauber-edge-time-jerrum} gives
    the time complexity
    \begin{align*}
            nm 
            \Delta^{O(1/\eps)}
            \left(
            n
            (\log n + \log \Delta^{1/\eps})
            + \log n
            \right)
            +
            O(m+n)
            =
            \frac{1}{\eps}
            \Delta^{O(1/\eps)}
            n^2m\log n.
    \end{align*}
    
\end{proof}

\section{Fast Algorithms for Planar and Bipartite Graphs}\label{sec:fast-algorithms}
In this section, we prove the following result: the existence of an algorithm that approximates the number of perfect matchings for a class $\mathcal{H}$ of graphs implies an algorithm that approximately samples from the Gibbs distribution over matchings in time $\mathrm{poly}(nm) \cdot \log \lambda$. 
This significantly improves the dependence on $\lambda$, reducing it from exponential (as in \Cref{thm:bounded-degree-algorithm}) to logarithmic.
\begin{theorem}\label{thm:perfect-reduction}
    Let $\mathcal H$ be a subgraph-closed class of graphs and $Q_{\delta'}$ be a parametrized algorithm
    that counts the number of perfect matchings with multiplicative error $\delta$
    for graphs in $\mathcal H$.
    Given the error parameter $\delta>0$ and the parameter of Gibbs distribution $\lambda>0$,
    let \[t\defeq \frac{n^2}{2}\left( m\log(1+\lambda) + \log\frac{3}{\delta}\right), \quad \delta'\defeq \frac{\log(1+\frac\delta 3)}{t} = \frac{\delta}{O\left(n^2 \left(m\log\lambda + \log\frac{1}{\delta}  \right) \right)}.\]
    Then, there exists an randomized algorithm $A$ that samples from the Gibbs
    distributions over matchings of graphs in $\mathcal H$ such that for any graph $G=(V, E)\in \mathcal H$, its output is within TV distance $\delta$ of $\mu_{E;\lambda,G}$, and makes $O\left(n^2(m\log \lambda+\log\frac1\delta)\right)$ calls
    of $Q_{\delta'}$ and $m$ calls of $Q_{\frac{\delta}{3m}}$
    on subgraphs of $G$.
\end{theorem}
As a corollary, if the running time of $Q_{\delta'}$ is $\mathrm{poly}(nm,{\delta'}^{-1})$, the running time of $A$ is also $\mathrm{poly}(nm,\delta^{-1})$,
and $A$ is a fully-polynomial time randomized approximation scheme (FPRAS).
\Cref{thm:perfect-reduction} is proved in \Cref{subsec:2-vertex-Glauber,subsec:mixing-time-of-gp}, and we use it to first prove \Cref{thm:planar-algorithm,thm:bipartite-algorithm}.

We first consider the planar case.
Although it is $\# P$-hard to count the number of perfect matchings on general graphs~\cite{valiant1979complexity},
it can be done on planar graphs in polynomial time:
\begin{lemma}[\cite{kasteleyn1963dimer}]\label{lem:fkt}
    There is an algorithm that computes the number of perfect matchings of a planar graph $G=(V, E)$ in time $O(n^3)$.
\end{lemma}

\begin{proof}[Proof of \Cref{thm:planar-algorithm}]
    Consider the algorithm in \Cref{thm:perfect-reduction} with the error parameter $\delta = 1/n^2$
    and $\lambda = 2(\varepsilon/2)^{-1} (4\Delta)^{1/\varepsilon}$.
    Then by \Cref{thm:perfect-reduction,lem:fkt}, the running time of the algorithm is
    \begin{align*}
        O\left(\frac{1}{\varepsilon}n^5m\log\Delta\right).
    \end{align*}
    By \Cref{thm:matching-approximation}, its approximation ratio is at least
    \[
        \left(1 - \frac{1}{n^2}\right)\left(1-\frac{\varepsilon}{2}\right) > 1-\varepsilon.
    \]
    By \Cref{thm:sensitivity}, the sensitivity is
    \[
        \left(1 - \frac{1}{n^2}\right) \cdot O(\lambda\Delta ) + \frac{1}{n^2} \cdot n
        = \Delta^{O(\varepsilon^{-1})}.
    \]
\end{proof}
To handle the bipartite case, we use the following lemma that shows the existence of a fully-polynomial time randomized approximation scheme (FPRAS) for the number of perfect matchings in bipartite graphs.
\begin{lemma}[\cite{bezakova2008accelerating,jerrum1989approximating}]\label{lem:bipartite-perfect}
    For any $\delta>0$, there is an algorithm $Q_\delta$ that outputs the number of perfect matchings of a bipartite graph within multiplicative error $\delta$
    and runs in time $O(n^{7}\log^4n + \delta^{-2}n^{6} \log^{5}n)$.
\end{lemma}
\begin{proof}[Proof of \Cref{thm:bipartite-algorithm}]
    Consider the algorithm in \Cref{thm:perfect-reduction} with the error parameter $\delta = 1/n^2$
    Then by \Cref{thm:perfect-reduction,lem:bipartite-perfect}, we need to
    call $O\left(\frac{1}{\varepsilon}n^2m\log\Delta\right)$ times of $Q_{\delta'}$, where
    \[
    \delta' = 
    \frac{1}{O\left(\frac1\varepsilon n^4 \left(m\log\Delta + \log n\right)\right)},
    \]
    and $m$ times of $Q_{\frac{\delta}{3m}}$ on subgraphs of $G$.
    Overall, the running time is
    \begin{align*}
        O\left(\frac{1}{\varepsilon^3}n^{16}m^3\log^3\Delta\log^5n\right).
    \end{align*}
    By \Cref{thm:matching-approximation}, its approximation ratio is at least
    \[
        \left(1 - \frac{1}{n^2}\right)\left(1-\frac{\varepsilon}{2}\right) > 1-\varepsilon.
    \]
    By \Cref{thm:sensitivity}, the sensitivity is
    \[
        \left(1 - \frac{1}{n^2}\right) \cdot O(\lambda\Delta ) + \frac{1}{n^2} \cdot n
        = \Delta^{O(\varepsilon^{-1})}.
    \]
\end{proof}

\subsection{The \texorpdfstring{$2$}{2}-vertex Glauber Dynamics}\label{subsec:2-vertex-Glauber}
The key idea behind our algorithm is to sample the \emph{matched vertices} instead of sampling directly from the Gibbs distribution over matchings. We then use the counting algorithm—whose existence is guaranteed by the theorem—to sample the matched edges.
We start with defining the Gibbs distribution of matched vertices.
\begin{definition}[Gibbs distribution of matched vertices]\label{def:gibbs-vertex}
    Let $G=(V, E)$ be a graph, $\lambda>0$ be a parameter and $V'\subseteq V$ be a subset of vertices.
    We define the \emph{Gibbs distribution of vertices} as the distribution on $2^{V'}$ such that, for any $U\subseteq V'$, we have
    \begin{align*}
        \mu_{V';\lambda,G}(U)
        \defeq
        \frac1Z
        \sum_{\substack{M\in \mathcal M(G) \\ V(M)\cap V'=U}} \lambda^{|M|},
    \end{align*}
    where $Z\defeq \sum_{M\in\mathcal M(G)}\lambda^{|M|}$ is the partition function.
\end{definition}
Here we slightly abuse the notation to use $\mu$ for both the Gibbs distribution of matched edges and vertices,
and we distinguish them by whether the first parameter in the subscript is an edge or vertex set.

\newcommand{\gp}{\mathcal P_{G}}
\newcommand{\mv}{\mu_{V; \lambda, G}}
The following Markov process is what we use to sample from $\mu_{V;\lambda,G}$.
\begin{definition}[$2$-vertex Glauber dynamics]\label{def:glauber-vertex}
    Given a graph $G=(V,E)$, a parameter $\lambda > 0$ we define the \emph{$2$-vertex Glauber dynamics} $\gp$ as a Markov chain on $\set{U\subseteq V\mid \mv(U)> 0}\subseteq 2^V$ such that
    \begin{align*}
        \gp(U\to U')
        \defeq
        \begin{cases}
            \frac{2}{n(n-1)}
            \frac{\mv(U'')}{\mv(U)+\mv(U')}
            &
            \text{if }|U\mathbin \triangle U'| = 2,
            \\
            \sum_{U''\in 2^V, |U\mathbin \triangle U''|=2}
            \frac{2}{n(n-1)}
            \frac{\mv(U)}{\mv(U)+\mv(U'')}
            &
            \text{if }U=U',
            \\
            0
            &
            \text{otherwise}.
        \end{cases}
    \end{align*}
\end{definition}
Intuitively, for the current state $U\subseteq V$, the dynamics does the following.
\begin{enumerate}
    \item Uniformly choose a random pair of vertices $\set{v_1, v_2}$ from $V$.
    \item Let $U' =  U\mathbin \triangle \set{v_1, v_2}$, then transit to $U'$ with probability $\frac{\mv(U')}{\mv(U)+\mv(U')}$ and stay at $U$ otherwise.
\end{enumerate}
The reason for using the $2$-site Glauber dynamics instead of the $1$-site version is to ensure the connectivity of the chain.
We first verify that this chain actually converges to $\mv$.
\begin{lemma}
    For any graph $G=(V,E)$, the following detailed balance condition hold:
    \begin{enumerate}
        \item 
        $\gp$ is connected;
        \item For any $U, U'\in 2^V$:
        $\mv(U) \gp(U\to U') = \mv(U') \gp(U'\to U)$.
    \end{enumerate}
    As a corollary, $\gp$ converges to $\mv$ (see, e.g.,~\cite{levin2017markov}).
\end{lemma}
\begin{proof}
    We first prove the connectivity.
    For any $U, U'\in 2^V$, with non-zero probabilities in $\mv$,
    we need to find a sequence $U_0 = U, U_1, U_2,\ldots, U_k=U'\in 2^V$ such that
    $\gp(U_j\to U_{j+1}) > 0$ for all $0\le j \le k-1$.
    
    Since $\mv(U),\mv(U') > 0$, there exist $M, M'\in \mathcal M(G)$ such that $V(M)=U, V(M')=U'$.
    By deleting the edges in $M\setminus M'$ one by one then adding the edges in $M'\setminus M$ one by one
    we get a sequence $M=M_0, M_1, \ldots, M=M_k\in \mathcal M(G)$ such that each neighboring pair is different by a single edge.
    Then we have a sequence $U=V(M), V(M_1),\ldots, V(M_k), U'=V(M')$ such that
    \begin{enumerate}
        \item $\mv(V(M_j)) > 0$ for all $0\le j\le k$;
        \item $|V(M_j)\mathbin \triangle V(M_{j+1})| = 2$ for all $0\le j\le k-1$.
    \end{enumerate}
    Then by \Cref{def:glauber-vertex}, we have $\gp(V_j\to V_{j+1}) > 0$ for all $0\le j \le k-1$, proving the connectivity.

    Then we consider the second part. If either $\mv(U)$ or $\mv(U')$ is zero then the equation trivially holds, 
    so we only need to consider the case that both of them have positive probability, which can be verified by
    direct calculation.
\end{proof}

With the convergence guarantee, we can start with an arbitrary vertex set with positive probability in $\mv$,
say $\varnothing$,
and run $\gp$ for sufficiently many times, then the resulting distribution would be close to the limit $\mv$.
To obtain \Cref{thm:perfect-reduction}, we need to prove that
\begin{enumerate}
    \item each step of $\gp$ can be simulated by polynomially many calls of $Q_{\delta'}$;
    \item it takes $\mathrm{poly}(nm,\delta)$ many steps of $\gp$ to converge to $\mv$ within TV distance $\delta$.
\end{enumerate}

The key component of this section is proving the second claim, which is presented in the next subsection. 
The proofs of the remaining claims are postponed until the end of this section.
\subsection{Mixing Time of 2-Site Glauber Dynamics }\label{subsec:mixing-time-of-gp}
In this section, we analyze the mixing time of the $2$-site Glauber dynamics $\gp$.
\subsubsection{Spectral Independence and Mixing Time}
\newcommand{\mk}{\mathcal P}

The mixing time of Glauber dynamics is controlled by a property of the stationary distribution called \emph{spectral independence}.
There are multiple variants of the definition, and here we follow the definition in \cite{chen2022localization}.
\begin{definition}[Spectral independence]
    Let $\mu$ be a distribution on $2^\Omega$ for some finite set $\Omega$. We say a event $\mathcal F$ is a \emph{pinning}
    if it is in the form of
    \[
    \mathcal F = \set{A\in 2^\Omega\mid A\cap B_- = \varnothing, B_+\subseteq A}
    \]
    for some $B_+, B_-\in 2^\Omega. 
    $
    For a pinning $\tau$, we denote the conditional distribution of $\mu$ on $\tau$ by $\mu^\tau$.
    We define the \emph{influence matrix} $\Psi^\tau_\mu (i, j) \in \mathbb R^{\Omega\times\Omega}$ of $\mu$ under a pinning $\tau$ as
    \[
        \Psi^\tau_\mu (i, j) \defeq
        \begin{cases}
        \Pr[A\sim \mu^\tau]{j\in A \mid i\in A}
        -
        \Pr[A\sim \mu^\tau]{j\in A \mid i\notin A}
        &
        \text{if }
        \Pr[A\sim \mu^\tau]{i\in A}
        \Pr[A\sim \mu^\tau]{i\notin A}
        >0,
        \\
        0 & \text{otherwise}.
        \end{cases}
    \]
    for all $i, j\in \Omega$.
    We say $\mu$ is $\eta$-spectral independent if for any pinning $\tau$, we have 
    $\rho(\Psi_\mu^\tau)\le \eta$,
    where $\rho$ denotes spectral norm.
\end{definition}
The following lemma characterizes the relation between spectral independence and mixing time.
The spectral gap of a Markov chain $\mathcal P$ is defined as $1-\lambda_2(P)$, where $\lambda_2(\mathcal P)$ is the second largest eigenvalue of $\mathcal P$.
\begin{lemma}[Theorem 24 in~\cite{chen2022localizationfull}]
    Let $\mu$ be an $\eta$-spectral independent distribution on $2^\Omega$ for some finite set $\Omega$.
    Then the $l$-site Glauber dynamics $\mathcal P$ has spectral gap 
    \[\lambda(\mathcal P)\ge \prod_{j=0}^{|\Omega|-l-1}\left( 1-\frac{\eta}{|\Omega|-j}, \right)\]
    and mixing time
    \[
    t_{\mathrm{mix}}(\mathcal P,\delta)
    \le \frac1{\lambda}\left( \log\frac{1}{\mu_{\min}} + \log\frac1{\delta} \right),
    \]
    where $\mu_{\min} = \min_{A\in 2^\Omega, \mu(A)\neq 0}\mu(A)$ is the minimum non-zero probability over $\Omega$.
\end{lemma}

The following shows that spectral independence is bounded by Wasserstein distance. 
\begin{lemma}[Proposition 4.1 in~\cite{chen2023near}]
    Let $\mu$ be a distribution on $2^\Omega$ for some finite set $\Omega$ and $\tau$ be an arbitrary pinning, then
    \[
    \rho(\Psi_\mu^\tau) \le \max_{\substack{i\in \Omega\\ 0<\mu^\tau(i\gets +)<1}} \W{\mu^{\tau, i\gets +}, \mu^{\tau, i\gets -}},
    \]
    where the metric used in the Wasserstein distance is the Hamming distance.
\end{lemma}
\begin{corollary}\label{lem:coupling-to-spectral}
    Let $\mu$ be a distribution on $2^\Omega$ for some finite set $\Omega$.
    If for any pinning $\tau$ and $i\in \Omega$ such that $0<\mu^\tau(i\gets +)<1$ there is
    \[
    \W{\mu^{\tau, i\gets +}, \mu^{\tau, i\gets -}} \le \eta,
    \]
    then $\mu$ is $\eta$-spectral independent.
\end{corollary}

\subsubsection{Recursive Coupling}
By \Cref{lem:coupling-to-spectral}, to analyze the mixing time of $\gp$ defined in \Cref{def:glauber-vertex}, it suffices to 
analyze the Wasserstein distance between pinnings as in the condition of \Cref{lem:coupling-to-spectral}.
To achieve this goal, we adopt a similar procedure to that in \Cref{subsec:matching-sensitivity}.
A crucial difference from the argument in \Cref{subsec:matching-sensitivity} is that we need to analyze the Wasserstein distances under all different vertex pinnings.

We start by defining upper bounds of Wasserstein distances as in \Cref{def:kappa}.
\begin{definition}[$\kappa^V_{s, \Delta, \lambda}$]\label{def:kappa-V}
    For each $s, \Delta\in \mathbb Z_{>0}$, $\lambda>0$, we define
    \begin{align*}
    \kappa^V_{s, \lambda}
    &\defeq
    \max_{\substack{G=(E, V)\\i=(u, v)\in E\\i\text{ is pendant}\\\Delta_G(u)=1\\ \tau:\text{ vertex pinning}}}
    \Wv{\mu_{V-u;\lambda, G}^{\tau,i\gets +},\mu_{V-u;\lambda, G}^{\tau ,i\gets -}},
    \end{align*}
    where $G=(V,E)$ ranges over all graphs with $|E|\le s$.
\end{definition}

\newcommand{\kv}[1]{\kappa^V_{#1, \lambda}}
Notice that there is no restriction on the maximum degree of $G$, in contrast with \Cref{def:kappa}, since the calculation here
doesn't need a bound like \Cref{lem:marginal-bound},
as we will see later in this section.
The following is a counterpart to \Cref{lem:kappa-recursion}.
\begin{lemma}\label{lem:kappa-v-recursion}
For all $s\in \mathbb Z_{\ge 1}$ and $\lambda>0$ we have
\[
    \kv{s+1}
    \le \max\set{\kv{s}, 1}.
\]
\end{lemma}

In the rest of this section, we prove \Cref{lem:kappa-v-recursion}.
We need the following lemma similar to \Cref{lem:to-kappa-edge}.
\newcommand{\mvgp}[3]{\mu_{#1 ; \lambda, #2}^{\substack{\tau \\ #3 }}}
\begin{lemma}\label{lem:to-kappa-v}
    Let $G=(V, E)$ be a graph with $|E|\le s$, $i = (u, v) \in E$ be a pendant edge with $\Delta_G(u)=1$, 
    $j\in E(v)$ be an edge other than $i$, and $\tau$ be an arbitrary vertex pinning.
    Then
    \begin{enumerate}
        \item 
        $\Wv{\mvgp {V-u} G {i\gets +}, \mvgp{V-u} G {v\gets -\\i\gets -}} = 1$;
        \item 
        $\Wv{\mvgp {V-u} G {i\gets +}, \mvgp{V-u} G {j\gets +\\i\gets -}} \le \kv{s-1}$.
    \end{enumerate}
\end{lemma}
\begin{proof}
    For the first part of the lemma, denote $E(v) - i$ by $F$.
    An application of \Cref{lem:pinning-deletion} gives
    \begin{align*}
        \mvgp {V-u} {G           } {i\gets +}
        =
        \mvgp {V-u} {G\setminus F} {i\gets +}
        \quad \text{and} \quad
        \mvgp {V-u} {G           } {v\gets - \\ i\gets -}
        =
        \mvgp {V-u} {G\setminus F} {i\gets -}.
    \end{align*}
    Since $\tau$ contains no pinning on $u, v$ by definition, there is a bijection:
    \begin{align*}
        \set{M\in \mathcal M(G\setminus F) \mid M \text{ satisfies } \tau, i\in M} \ni M
        & \leftrightarrow
        M - i \in \set{M\in \mathcal M(G\setminus F) \mid M \text{ satisfies } \tau, i\notin M}.
    \end{align*}
    Hence, we have $\Wv{\mvgp {V-u} G {i\gets +}, \mvgp{V-u} G {v\gets -\\i\gets -}} = 1$.

    For the second part, denote $F \defeq E(v) - i - j$. Notice that 
    \begin{align*}
        \mvgp {V-u} G {i\gets +} (j\gets-, F\gets -) = 1,\qquad
        \mvgp {V-u} G {j\gets + \\ i\gets -} (F\gets - ) = 1,
    \end{align*}
    so we have
    \begin{align*}
        \mvgp {V-u} {G\setminus F} {j\gets - \\ i\gets + } = \mvgp {V-u} G {i\gets +}, \qquad
        \mvgp {V-u} {G\setminus F} {j\gets + \\ i\gets - } = \mvgp {V-u} G {j\gets + \\ i\gets -}.
    \end{align*}
    Moreover, since either the pinning $i\gets +$ or $j\gets +$ makes $v$ to appear in the matching, we have
    \begin{align*}
        \mvgp {V-u} {G\setminus F} {j\gets - \\ i\gets + }(v\gets +) =
        \mvgp {V-u} {G\setminus F} {j\gets + \\ i\gets - }(v\gets +) = 1.
    \end{align*}
    So we may remove $v$ from the vertex set without affecting the Wasserstein distance between the distributions.
    That is
    \begin{align*}
         \Wv{\mvgp { { V-u}} G {i\gets +}, \mvgp{V-u} G {j\gets +\\i\gets -}}
        = 
         \Wv{
             \mvgp {V-u-v} {G\setminus F} {j\gets - \\ i\gets + },
             \mvgp {V-u-v} {G\setminus F} {j\gets + \\ i\gets - }         
             }.
    \end{align*}
    For both distributions on the right hand side, $i$ is a pendant edge whose only neighbour is $j$.
    Since $j$ is also pinned, the pinning on $i$ does not affect the distribution on $V-u-v$, and we may
    remove $i$ from both the distributions without changing them. Finally we have
    \begin{align*}
         \Wv{\mvgp { { V-u}} G {i\gets +}, \mvgp{V-u} G {j\gets +\\i\gets -}}
        = 
         \Wv{
             \mvgp {V-u-v} {G\setminus F - i} {j\gets -},
             \mvgp {V-u-v} {G\setminus F - i} {j\gets +}         
             }.
    \end{align*}
    Since $j$ is a pendant edge in $G\setminus F - i$, a graph with at most $s-1$ edges, the last term is bounded by $\kv{s-1}$, 
    concluding the proof.
\end{proof}

\begin{proof}[Proof of \Cref{lem:kappa-v-recursion}]
\renewcommand{\m}{\mu_{E;\lambda, G}}
\renewcommand{\mp}{ \mvgp{V-u}{G}{            i\gets + }}
\renewcommand{\mm}{ \mvgp{V-u}{G}{            i\gets - }}
\renewcommand{\mmp}{\mvgp{V-u}{G}{j\gets + \\ i\gets - }}
\renewcommand{\mmm}{\mvgp{V-u}{G}{v\gets - \\ i\gets - }}
\renewcommand{\pmp}{\mvgp{V-u}{G}{            i\gets - }(j\gets +)}
\renewcommand{\pmm}{\mvgp{V-u}{G}{            i\gets - }(v\gets -)}
\renewcommand{\se}{\sum_{j\in E(v)}}
Let $G=(V, E)$ be a graph and $i=(u, v)\in E$ be an pendant edge that fit into the definition of $\kv s$.
By \Cref{lem:expansion} we have
\begin{align*}
    \mp - \mm
    &=
    \mp - \se\pmp\mmp -\pmm\mmm
    \\&=
    \se \pmp \left( \mp - \mmp \right)
    \\&\phantom{=}+
    \pmm\left( \mp - \mmm \right).
\end{align*}
By \Cref{lem:wd-sum}, we have
\begin{align*}
    \Wv{\mp , \mm}
    &\le
    \se \pmp \Wv{ \mp , \mmp}
    \\&\phantom{=}+
     \pmm\Wv{\mp , \mmm}
    \\&\le \max\set{\Wv{ \mp , \mmp},\Wv{\mp , \mmm}}.
\end{align*}
By \Cref{lem:to-kappa-v}, we have
\begin{align*}
    \Wv{\mp , \mm}
    \le \max\set{\kv{s-1}, 1}.
\end{align*}
Taking the supremum over all possible $G$ and $i$ that fit the definition of $\kv s$ finishes the proof.
\end{proof}

Finally, we prove \Cref{thm:perfect-reduction}.
\begin{proof}[Proof of \Cref{thm:perfect-reduction}]
    For any $G=(V, E)$, $\lambda > 0$, $v\in V$ and vertex pinning $\tau$, we need to bound
    \begin{align}\label{eq:wd-v-independence}
        \Wv{\mvgp{V}{G}{v\gets +}, \mvgp{V}{G}{v\gets -}}.
    \end{align}
    We have the decomposition:
    \begin{align*}
        \mvgp{V}{G}{v\gets +} - \mvgp{V}{G}{v\gets -}
        =
        \sum_{i\in E(v)}
        \mvgp{V}{G}{v\gets +}(i\gets +)
        \left(
        \mvgp{V}{G}{v\gets +} - \mvgp{V}{G}{v\gets -}
        \right).
    \end{align*}
    Then by \Cref{lem:wd-sum}, we have
    \begin{align*}
    \Wv{
        \mvgp{V}{G}{v\gets +}, \mvgp{V}{G}{v\gets -}
    }
    &\le
        \sum_{i\in E(v)}
        \mvgp{V}{G}{v\gets +}(i\gets +)
        \Wv{
        \mvgp{V}{G}{i\gets +}, \mvgp{V}{G}{v\gets -}
        }
    \\&\le
        \max_{i\in E(v)}
        \Wv{
        \mvgp{V}{G}{i\gets +}, \mvgp{V}{G}{v\gets -}
        }.
    \end{align*}
    By \Cref{lem:pinning-deletion}, denoting $E(v)-i$ by $F$, we have
    \begin{align*}
        \Wv{
        \mvgp{V}{G}{i\gets +}, \mvgp{V}{G}{v\gets -}
        }
        &=
        \Wv{
        \mvgp{V}{G\setminus F}{i\gets +}, \mvgp{V}{G\setminus F}{i\gets -}
        }
        \\&=
        1+
        \Wv{
        \mvgp{V-v}{G\setminus F}{i\gets +}, \mvgp{V-v}{G\setminus F}{i\gets -}
        }
        \\&\le 1+\kv{|E|},
    \end{align*}
    since $i$ is pendant in $G\setminus F$.
    Finally, \Cref{lem:kappa-v-recursion} together with the fact that $\kv 1 =1$ indicate that $\kv s\le 1$ for all $s\in \mathbb Z_{\ge 1}$.
    So we know the Wasserstein distance in \eqref{eq:wd-v-independence} is universally bonded by $2$.
    Then by \Cref{lem:coupling-to-spectral},
    $\mu_{V;\lambda, G}$ is $2$-spectral independent, and the $2$-vertex Glauber dynamics $\gp$ has the following bound on spectral gap:
    \begin{align*}
        \lambda(\gp) \ge \prod_{j=0}^{n-3} \left( 1 - \frac{2}{n-j} \right) > \frac 2 {n^2}.
    \end{align*}
    Notice that when $\lambda > 1$, the minimum non-zero probability in the distribution $\mu_{V;\lambda, G}$ is
    \begin{align*}
        \mu_{V;\lambda, G}(\varnothing) =
        \frac{1}{\sum_{M\in \mathcal M(G)}|\lambda|^{|M|}}
        \le 
        \frac{1}{\sum_{M\in \{0,1\}^E}|\lambda|^{|M|}}
        =\frac{1}{(1+\lambda)^{m}}.
    \end{align*}
    Then we get the mixing time of $\gp$ as
    \begin{align*}
        t_{\mathrm{mix}}(\gp, \delta) \le \frac{n^2}{2}\left(m\log(\lambda +1) +\log\frac1\delta\right)  =  O\left(n^2\left(m\log\lambda + \log\frac{1}{\delta}\right)\right).
    \end{align*}

    In the following discussion, we fix $t = t_{\mathrm{mix}}(\gp,\frac{\delta}{3})$.
    Since $Q_{\delta'}$ is an approximation algorithm, we are unable to run the Glauber dynamics $\gp$ directly.
    \newcommand{\pgp}{\mathcal P'_G}
    Instead, we use the following Markov chain $\pgp$ to simulate it. If the current state is $U\in 2^V$, then
    $\pgp$ does the following.
    \begin{enumerate}
        \item Uniformly sample a pair of distinct vertices $\set{v_1,v_2}$ from $V$.
        \item Let $U' = U\mathbin\triangle\set{v_1, v_2}$,
        transit to $V\mathbin\triangle \set{v_1, v_2}$ with probability $\frac{Q_{\delta'}(G[U'])}{Q_{\delta'}(G[U']) + Q_{\delta'}(G[U])}$ and stay at $U$ otherwise.
    \end{enumerate}
    Under this definition, the following two facts hold.
    \begin{enumerate}
        \item $\pgp(U\to U') = 0$ if $\mv(U') = 0$ since in this case $G[U']$ has no perfect matching and our approximation algorithm $Q_{\delta'}$ always outputs $0$.
        \item For any $U, U'$ in the support of $\mv$, we have 
        \[
        |\pgp(U\to U') - \gp(U\to U')|\le \frac{2\delta'}{1-\delta'}\gp(U\to U).
        \]
    \end{enumerate}
    It takes $2$ calls of $Q_{\delta'}$.
    Now, for any vector $x\in\mathbb R^{2^V}$, we have
    \begin{align*}
    \Vert x\pgp - x\gp\Vert_1
    &=
    \sum_{U'\in 2^V}\sum_{U\in 2^V} x(U)|\pgp(U\to U') - \gp(U\to U')|
    \\&\le
    \frac{2\delta'}{1-\delta'}
    \sum_{U'\in 2^V}\sum_{U\in 2^V} x(U) \gp(U\to U').
    \\&=
    \frac{2\delta'}{1-\delta'}\Vert x\Vert_1.
    \end{align*}
    Moreover, for any $x\in \mathbb R^{2^V}$, we have
    $\Vert x\gp\Vert_1\le\Vert x\Vert_1$ since the operator norm of $\gp$ is $1$.
    Then, the total variation error after running $\pgp$ for $t$ times with an initial distribution $\mu_0$ is
    \begin{align*}
    \Vert \mu (\pgp)^t - \mv \Vert_1
    &=
    \Vert \mu(\gp + (\pgp-\gp))^t - \mv \Vert_1
    \\&=
    \Vert \mu\gp^t  +  \mu\left((\gp + (\pgp-\gp))^t - \gp^t\right) - \mv\Vert_1
    \\&\le
    \Vert \mu\gp^t -\mv\Vert_1  + \Vert \mu\left((\gp + (\pgp-\gp))^t - \gp^t\right) \Vert_1
    \\&\le
    \frac{\delta}{3} + (1+\delta')^t -1
    \le
    \frac{\delta}{3} + e^{t\delta'} -1 = \frac23\delta.
    \end{align*}

    By the standard reduction from counting to sampling,
    after sampling a vertex set $U$ from $\mv$, we can call $Q_{\frac{\delta}{3m}}$ for $m$ times
    on subgraphs of $G[U]$ to sample the set of matched edges within total variation distance $\frac{\delta}{3}$.
    Finally we have the overall total variation  error $\delta$.
    
    Since it takes $2$ calls of $Q_{\delta'}$ to  perform one step of $\gp$ as discussed above, in total we need 
    $O(n^2\left(m\log\lambda + \log\frac{1}{\delta})\right)$ calls of $Q_{\delta'}$
    and $m$ calls of $Q_{\frac{\delta}{3m}}$.
\end{proof}

\section{Low-Sensitivity Algorithms for General Graphs}\label{sec:sparsification}

The algorithm presented in \Cref{sec:sensitivity} has high sensitivity and running time when the maximum degree is large. 
In this section, we address this issue by sparsifying the input graph into one with bounded degree while approximately preserving the matching size. 
We then prove \Cref{thm:unbounded-degree-algorithm} by applying a version of \Cref{thm:bounded-degree-algorithm} obtained by using a slightly different bound on the mixing time of the Gibbs distribution to the resulting sparsified graph.

We begin in \Cref{subsec:lp-relaxation} by showing how to solve the LP relaxation of the maximum matching problem with low sensitivity, which is a key step in constructing the sparsifier.
Then in \Cref{subsec:sparsification}, we describe the sparsification procedure using the obtained LP solution, and complete the proof of \Cref{thm:unbounded-degree-algorithm}.

\subsection{LP Relaxation}\label{subsec:lp-relaxation}
Our sparsification algorithm is based on linear programming, and we start with explaining the LP relaxation of the maximum matching problem.
For an undirected graph $G = (V,E)$, the matching polytope $\mathscr M_G$ of $G$ is the convex hull of the indicator vectors of matchings in $G$. 
Consider the following polytope, defined by constraints ensuring that the total weight of edges incident to any vertex in the matching is at most one:
\[
    \mathscr P_G := \left\{\sum_{e \in E(v)}x(e) \leq 1, \forall v\in V\right\} \cap \mathbb R_{\geq 0}^E.
\]
It is standard that $\mathscr M_G = \mathscr P_G$ holds when $G$ is bipartite.
When $G$ is non-bipartite, we need to further consider odd-set constraints. 
For $x \in \mathbb R_{\geq 0}^E$ and $B\subseteq V$, let $x(B) := \sum_{e \in E(G[B])}x(e)$.
Then, we let
\[
    \mathcal O_G := \{B \subseteq V : |B|\geq 3\text{ and }|B|\text{ is odd}\}
\]
be the collections of odd sets.
Then it is known that we have
\[
    \mathscr M_G= \mathscr P_G \cap \left\{x(B) \leq \left\lfloor \frac{|B|}{2}\right\rfloor, \forall B \in \mathcal O_G\right\}.
\]
Hence, we can solve the maximum matching problem by maximizing $x(E)$ subject to $x \in \mathscr M_G$.

To solve the LP relaxation in a stable fashion, we introduce an entropy function $H:\mathbb R_{\geq 0}^F \to \mathbb R_{\geq 0}$ defined as 
\[
    H(x) = \sum_{e \in F} x(e) \log \frac{1}{x(e)}
\]
It is well known that $H$ is concave and $1$-strongly concave with respect to the $\ell_1$ norm in the simplex $\{x \in \mathbb R_{\geq 0}^F: x(F) \leq 1 \}$, where $x(F) = \sum_{e \in F}x(e)$.
Let $K_V$ denote the complete graph on the vertex set $V$.
For a parameter $\alpha \geq 0$, we consider the following entropy-regularized LP:
\begin{align}
    \begin{array}{lll}
        \text{maximize} & \displaystyle x(E) + \alpha \sum_{v \in V}H(\{x(e)\}_{e \in E(v)}) \\
        \text{subject to} & x \in \mathscr M_{K_V} \\
        & x \in \mathbb R_{\geq 0}^{\binom{V}{2}}
    \end{array}
    \label{eq:regularized-LP}
\end{align}
We consider the matching polytope over $K_V$, instead of $G$, to ensure that the constraints remain independent of the graph, which is crucial for the stability of the LP solution.
Note that $x(e)$ is defined even for $e \not \in E$.
Since the objective is a concave function and we can construct a separation oracle for the constraints, the problem can be solved in polynomial time using the ellipsoid method.

For a vector $x \in \mathbb R^{\binom{V}{2}}_{\geq 0}$ and an edge set $E \subseteq \binom{V}{2}$, let $x|_E \in \mathbb R^{\binom{V}{2}}_{\geq 0}$ denote the vector obtained by truncating $x$ to $E$, i.e., $x|_E(e) = x(e)$ if $e \in E$ and $x|_E(e) = 0$ otherwise.
\begin{lemma}\label{lem:lp-approximation}
    Let $x \in \mathbb R^{\binom{V}{2}}_{\geq 0}$ be the optimal solution to the regularized LP~\eqref{eq:regularized-LP} and let $x' = x|_E$.
    Then, we have 
    \[
        x'(E) \geq \nu(G) - \alpha \nu(G)\log \nu(G).
    \]
\end{lemma}
\begin{proof}
    Let $M$ be the maximum matching of $G$ and $\bm 1_{M}$ be its indicator vector.
    We have
    \[
        x(E) + \alpha \sum_{v \in V}H(\{x(e)\}_{e \in E(v)}) 
        \geq 
        \bm 1_{M}(E) + \alpha \sum_{v \in V}H(\{\bm 1_{M}(e)\}_{e \in E(v)})
        \geq |M|.       
    \]
    Hence, we have 
    \begin{align*}
        & x'(E) = x(E) \geq |M| - \alpha \sum_{v \in V}H(\{x(e)\}_{e \in E(v)}) 
        = |M| - \alpha \sum_{e \in E} x(e) \log \frac{1}{x(e)} \\
        & = |M| - \alpha |M| \sum_{e \in E} \frac{x(e)}{|M|} \log \left(\frac{|M|}{x(e)} \frac{1}{|M|}\right)
        = |M| - \alpha |M| \sum_{e \in E} \frac{x(e)}{|M|} \left(\log \frac{|M|}{x(e)} + \log \frac{1}{|M|}\right) \\
        & \geq |M| - \alpha |M| \log |M| - \alpha |M| \sum_{e \in E} \frac{x(e)}{|M|} \log \frac{1}{|M|} \\
        & \geq |M| - \alpha |M| \log |M|.
    \end{align*}
\end{proof}

\begin{lemma}\label{lem:lp-sensitivity}
    Let $G=(V,E)$ be a graph and let $\tilde G=(V,\tilde E= E-e)$ for some $e =(u,v)\in E$.
    Let $x$ and $\tilde x$ be the optimal solutions of the regularized LP~\eqref{eq:regularized-LP} for $G$ and $\tilde G$, respectively.
    Then, we have
    \[
        \sum_{v \in V} \left(\sum_{f \in E_G(v)}|\tilde x(f) - x(f)|\right)^2 \leq O\left(\frac{\log n}{\alpha}\right).
    \]
\end{lemma}
\begin{proof}    
    Let $h_G:\mathbb R_{\geq 0}^{\binom{V}{2}} \to \mathbb R_{\geq 0}$ be the objective function of \eqref{eq:regularized-LP}.
    Define $D := h_G(x) - h_G(\tilde x) + h_{\tilde G}(\tilde x) - h_{\tilde G}(x)$.
    First note that 
    \begin{align*}
        & D 
        = x(E) - \tilde  x(E) + \tilde  x(\tilde  E) - x(\tilde  E) \\
        & \quad + H(\{x(e)\}_{e \in E_G(u)}) + H(\{x(e)\}_{e \in E_G(v)})
        - (H(\{\tilde x(e)\}_{e \in E_G(u)}) + H(\{x(e)\}_{\tilde e \in E_G(v)})) \\
        & \qquad + H(\{\tilde x(e)\}_{e \in E_{\tilde G}(u)}) + H(\{\tilde x(e)\}_{e \in E_{\tilde G}(v)})
        - (H(\{x(e)\}_{e \in E_{\tilde G}(u)}) + H(\{x(e)\}_{e \in E_{\tilde G}(v)})) \\
        & \leq x(e) - \tilde x(e) + 8\log n
        \leq 9\log n.
    \end{align*}
    By strong concavity of the entropy function with respect to $\ell_1$, we have by Lemma 7.13 of \cite{kumabe2022lipschitz-arxiv} that 
    \begin{align*}
        & h_G(x)  - h_G(\tilde x) \geq \alpha \sum_{v \in V} \left(\sum_{e \in E_G(v)}|\tilde x(e) - x(e)|\right)^2 \\
        & h_{\tilde G}(\tilde x) - \tilde h_{\tilde G}(x) \geq \alpha \sum_{v \in V} \left(\sum_{e \in E_{\tilde G}(v)}|\tilde x(e) - x(e)|\right)^2.
    \end{align*}        
    Hence, we have
    \begin{align*}
        D 
        & \geq \alpha \sum_{v \in V} \left(\sum_{e \in E_G(v)}|\tilde x(e) - x(e)|\right)^2 + \alpha \sum_{v \in V} \left(\sum_{e \in E_{\tilde G}(v)}|\tilde x(e) - x(e)|\right)^2  \\
        & \geq \alpha \sum_{v \in V} \left(\sum_{e \in E_G(v)}|\tilde x(e) - x(e)|\right)^2,
    \end{align*}
    which implies
    \[
        \sum_{v \in V} \left(\sum_{e \in E(v)}|\tilde x(e) - x(e)|\right)^2
        \leq O\left(\frac{\log n}{\alpha}\right).
    \]
\end{proof}

\subsection{Sparsification}\label{subsec:sparsification}

To sparsify the graph preserving the matching size, we use a slight variant of the notion of degree sparsifier~\cite{chen2025entropy}.
For a graph $G=(V,E)$, a vertex $v \in V$, and a vector $x \in \mathbb R^E$, let $x(v) = \sum_{e \in E(v)}x(e)$.
\begin{definition}[Degree sparsifier]\label{def:degree-sparsifier}
    Let $G=(V,E)$ be a graph.
    For a fractional matching $x \in \mathscr M_G$, a subgraph $H = (V,F) \subseteq \mathrm{supp}(x)$ is an $(\Delta,\varepsilon)$-degree-sparsifier for $\Delta = \Delta(n,m,\varepsilon)$ of $G$ if $\Delta(v) \leq \Delta$ for every $v \in V$\footnote{In the original definition, we require that $|F| \leq s \|x\|_1$ for some $s = s(n,m,\varepsilon)$.} and there exists a fractional matching $x^{(H)} \in \mathscr M_G$ supported on $F$, called a certificate of $H$, such that
    \begin{itemize}
        \item $\|x^{(H)}\|_1 \geq(1-\varepsilon)\|x\|_1$.
        \item $x^{(H)}(v) \geq x(v)-\varepsilon$ for all $v\in V$.
        \item $x^{(H)}(B) \geq x(B)-\varepsilon |B|/3 \geq x(B)-\varepsilon \lfloor|B|/2\rfloor$ for all odd sets $B \in \mathcal O_G$.    
    \end{itemize}
\end{definition}
Note that the maximum matching size of an $(\Delta,\varepsilon)$-degree-sparsifier $H$  for a fractional matching $x \in \mathscr M_G$ is at least $(1-\varepsilon)\|x\|_1$.

\begin{lemma}[Lemma 6.15 of~\cite{chen2023entropy-arxiv}]\label{lem:chen}
    Let $G=(V,E)$ be a graph.
    For any fractional matching $x \in \mathscr M_G$ with $\|x\|_1 \geq 1$ and constant $c > 0$, the random subgraph $H$ obtained by sampling each edge $e$ independently with probability $p_e$ defined as
    \[
        p_e := \min\left\{1,\frac{x(e)}{\gamma}\right\}, \text{ where } \gamma := \frac{\varepsilon^2}{320 \max\{1,c\} \log n},
    \]
    is an $(O(\varepsilon^{-2}\log n),\varepsilon)$-degree-sparsifier with probability at least $1-n^{-c}$.
\end{lemma}
We remark on the maximum degree guarantee, which we modified from the original definition of a degree sparsifier and which is not addressed in \cite{chen2025entropy}.  
In Claim 6.17 of \cite{chen2023entropy-arxiv}, they define  
\[
x'(e) = \begin{cases}
x(e) & \text{if } x(e) \geq \gamma, \\
\gamma & \text{if } x(e) < \gamma \text{ and } e \in H, \\
0 & \text{if } x(e) < \gamma \text{ and } e \notin H,
\end{cases}
\]  
and show that $|x'(v) - x(v)| \leq \varepsilon$ with probability at least $1 - 2n^{-20\max\{1,c\}}$.  
Note that $x'(v)/\gamma$ is an upper bound on the degree of $v$ in $H$.  
Hence, this implies that $\Delta(v) \leq (x(v)+\varepsilon)/\gamma = O(1/\gamma)$ with high probability, and we obtain \Cref{lem:chen}.

Our sparsification algorithm is given in Algorithm~\ref{alg:sparsification}.
After computing $x \in \mathbb R^{\binom{V}{2}}$, we first truncate it so it is supported by $E$.
Then for each $e \in E$, we add $e$ to the sparsifier $H=(V,F)$ with probability $p_e$ independently from others.

\begin{algorithm}[t!]
  \caption{Sparsify$(G, \varepsilon, \alpha)$}
  \label{alg:sparsification}
  \begin{algorithmic}[1]
    \Procedure{Sparsify}{$G, \varepsilon, \alpha$}
      \State Solve LP~\eqref{eq:regularized-LP} with $\alpha$ and let $x \in \mathbb{R}^{\binom{V}{2}}$ be the obtained solution.
      \State Let $x' := x|_E \in \mathbb{R}^{\binom{V}{2}}$ be the fractional matching over $E$.
      \State $F \gets \emptyset$
      \For{$e \in E$}
        \State \text{Add $e$ to $F$ with probability $p_e$}
      \EndFor
      \State \Return $H = (V, F)$
    \EndProcedure
  \end{algorithmic}
\end{algorithm}

\begin{theorem}\label{thm:sparsification}
    Let $H= \Call{Sparsify}{G,\varepsilon,\alpha}$.
    Then, 
    \begin{itemize}
        \item With probability at least $1-n^{-2}$, the maximum degree of $H$ is $O(\varepsilon^{-2}\log n)$ and 
        \[
            \nu(H) \geq (1-\varepsilon)(\nu(G)-\alpha \nu(G)\log \nu(G)).
        \]
        \item  The sensitivity of \Call{Sparsify}{} is 
        \[
            O\left(\frac{\sqrt{n} \log^{3/2}n}{\varepsilon^2 \sqrt{\alpha}}\right).
        \]
    \end{itemize}
\end{theorem}
\begin{proof}
    We instantiate \Cref{lem:chen} with $c=2$, and we obtain that $H$ is an $O(\varepsilon^{-2})$-sparse $\varepsilon$-degree sparsifier with probability at least $1-n^{-2}$ by \Cref{lem:chen}.
    Then by \Cref{lem:lp-approximation} and the observation below \Cref{def:degree-sparsifier}, the maximum matching size of $H$ is at least $(1-\varepsilon)(\nu(G)-\alpha \nu(G)\log \nu(G))$.

    Next, we analyze the sensitivity of \Call{Sparsify}{}.
    Let $G=(V,E)$ be a graph and $\tilde G = (V,\tilde E = E-e)$ for an edge $e \in E$.
    Let $x$ and $\tilde x$ be the solutions of the regularized LP~\eqref{eq:regularized-LP} for $G$ and $G'$, respectively.
    Let $x' = x|_E$ and $\tilde x' = \tilde x|_{\tilde E}$ be the truncated LP solutions.
    Then by \Cref{lem:lp-sensitivity}, we have 
    \[
        \sum_{v \in V} \left(\sum_{f \in E_G(v)}|\tilde x'(f) - x'(f)|\right)^2
        = O\left(\frac{\log n}{\alpha}\right).
    \]
    Then, we have
    \begin{align*}        
        & \sum_{f \in E} |x'(f) - \tilde x'(f)|
        =
        \frac{1}{2}\sum_{v \in V}  \sum_{e \in E_G(v)} |x'(f) - \tilde x'(f)| \\
        & \leq 
        \frac{1}{2}\sqrt{\sum_{v \in V} \left(\sum_{f \in E_G(v)} |x'(f) - \tilde x'(f)|\right)^2} \sqrt{\sum_{v \in V}1^2} \tag{by Cauchy-Schwarz} \\
        & \leq O\left(\sqrt{\frac{n\log n}{\alpha} }\right).
    \end{align*}
    Let $H= \Call{Sparsify}{G,\varepsilon,\alpha}$ and $\tilde H= \Call{Sparsify}{\tilde G,\varepsilon,\alpha}$.
    In the construction of a degree sparsifier of \Cref{lem:chen}, we sample edges independently from others.
    Hence by \Cref{lem:lp-sensitivity}, we have 
    \begin{align*}
        & \We{E(H), E(\tilde H)}
        \leq 
        \frac{1}{2}\sum_{f \in \tilde E} \left|\Pr{f \in E(H)} - \Pr{f \in E(\tilde H)}\right| + \frac{1}{2}\Pr{f \in E} \\
        & \leq \frac{1}{2\gamma} \sum_{f \in E} |x'(f) - \tilde x'(f)| 
        = O\left(\frac{1}{\gamma}\sqrt{\frac{n\log n}{\alpha}}\right)
        = O\left(\frac{\sqrt{n} \log^{3/2} n}{\varepsilon^2 \sqrt{\alpha}}\right).
    \end{align*}
\end{proof}

As the maximum degree of the degree-sparsifier we construct is $O(\log n)$, we use \Cref{lem:glauber-edge-time-jerrum} instead of \Cref{lem:glauber-edge-time}.

\begin{proof}[Proof of~\Cref{thm:unbounded-degree-algorithm}]
    Consider the algorithm first apply \Cref{thm:sparsification} with $\alpha = O(\varepsilon/\log n)$, and then apply  the algorithm of \Cref{thm:bounded-degree-algorithm} on the resulting graph.
    However, we use \Cref{lem:glauber-edge-time-jerrum} instead of \Cref{lem:glauber-edge-time} for the analysis of the running time.

    As the resulting graph is $(O(\varepsilon^{-2}\log n),\varepsilon)$-degree sparsifier with probability at least $1-1/n^2$, the expected size of the output matching is 
    \[
        (1-\varepsilon) \left(1-\frac{1}{n^2}\right)  (1-\varepsilon)\left(\nu(G) - \frac{\varepsilon}{\log n} \nu(G)\log \nu(G)\right)   
        = (1-O(\varepsilon))\nu(G).
    \]
    The sensitivity is
    \[
        O\left(\frac{\sqrt{n}\log^2 n}{\varepsilon^{5/2} }\right) 
        \cdot 
        \left(\left(1-\frac{1}{n^2}\right) \left(\frac{\log n}{\varepsilon^2}\right)^{O(1/\varepsilon)} + \frac{1}{n^2} \cdot n \right)
        = \sqrt{n} \left(\frac{\log n}{\varepsilon}\right)^{O(1/\varepsilon)}.
    \]
    The running time bound of the algorithm is 
    \[
        O\left(nm \max\{1,\lambda\} \left(n(\log n + \log \max\{1,\lambda\}) + \log \frac{1}{\delta}\right)\right)
    \]
    for $\lambda = (\varepsilon^{-2}\log n)^{O(1/\varepsilon)}$ and $\delta=1/n^2$, which is 
    \[
        O\left(n^2 m (\varepsilon^{-1}\log n)^{O(1/\varepsilon)}\right).
    \]    
\end{proof}

\section{Lower Bounds}\label{sec:lower-bound}

In this section, we show that the exponential dependence on $1/\varepsilon$ in the parameter $\lambda$ of our Gibbs-distribution-based algorithm (\Cref{thm:bounded-degree-algorithm}) is unavoidable in order to achieve a $(1 - \varepsilon)$-approximation:
\begin{theorem}\label{thm:lower-bound}
    Let $\varepsilon > 0$ be sufficiently small and $\lambda \geq 0$.
    If $\E[M\sim \mu_{E;\lambda,G}]{|M|} \geq (1-\varepsilon)\nu(G)$ for any graph $G$, then we have $\lambda = 2^{\Omega(1/\varepsilon)}$.
\end{theorem}

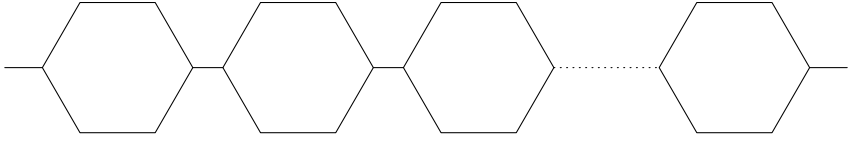
\begin{figure}
    \centering
    \begin{tikzpicture}[line join=round, line cap=round]
      \begin{scope}[shift={(0,0)}]
        \draw (0:1) -- (60:1) -- (120:1) -- (180:1) 
              -- (240:1) -- (300:1) -- cycle;
      \end{scope}
      \draw (-1,0) -- (-1.5,0);
      \draw (1,0) -- (1.4,0);
    
      \begin{scope}[shift={(2.4,0)}]
        \draw (0:1) -- (60:1) -- (120:1) -- (180:1) 
              -- (240:1) -- (300:1) -- cycle;
      \end{scope}
      \draw (3.4,0) -- (3.8,0);
    
      \begin{scope}[shift={(4.8,0)}]
        \draw (0:1) -- (60:1) -- (120:1) -- (180:1) 
              -- (240:1) -- (300:1) -- cycle;
      \end{scope}
      \draw[dotted] (5.8,0) -- (7.2,0);
    
      \begin{scope}[shift={(8.2,0)}]
        \draw (0:1) -- (60:1) -- (120:1) -- (180:1)
              -- (240:1) -- (300:1) -- cycle;
      \end{scope}
      \draw (9.2,0) -- (9.7,0);
    
    \end{tikzpicture}
    
    \caption{A chain of hexagons}\label{fig:hexagon}
\end{figure}

Let $H$ be the hexagon graph with two distinguished vertices $s$ and $t$, connected by two internally vertex-disjoint paths of length $3$.  
Formally, let  
\[
    V(C_k) = \{s, t, u_1, v_1, u_2, v_2\}
\]  
and  
\[
    E(H) = \left\{(s, u_i), (u_i, v_i), (v_i, t) \mid i \in [2]\right\}.
\]
For positive integer $\ell$, we construct a chain-of-hexagon graph $G_\ell$ by chaining together $\ell$ copies of $H$, denoted $H^{(1)}, \ldots, H^{(\ell)}$, and then connecting them with additional edges to form a path-like structure from a new source $s$ to a new target $t$.
Specifically, let  
\[
    V(G_\ell) = \{s, t\} \cup \bigcup_{j=1}^\ell V(H^{(j)}),
\]  
where each $H^{(j)}$ is an independent copy of $H$, with endpoints $s^{(j)}$ and $t^{(j)}$. 
The edge set is defined as  
\[
    E(G_\ell) = \{(s, s^{(1)}), (t^{(\ell)}, t)\} \cup \bigcup_{j=1}^{\ell-1} \{(t^{(j)}, s^{(j+1)})\} \cup \bigcup_{j=1}^\ell E(H^{(j)}).
\]
See \Cref{fig:hexagon} for a schematic illustration.
We note that this graph has been already appeared in \cite[Section 5.4]{jerrum2003counting}.

\begin{lemma}\label{lem:chain-of-hexagons}
    We have the following:
    \begin{itemize}
    \item The graph $G_\ell$ is bipartite and planar.
    \item $G_\ell$ has a unique perfect matching with size $\nu(G_\ell) = 2 + 2\ell$.
    \item The number of matchings of size $\nu(G_\ell)-1$ is at least $2^\ell$.
    \end{itemize}
\end{lemma}
\begin{proof}
    The first claim is clear.
    
    For the second claim, observe that the edge set  
    \[
        M := \{(s, s^{(1)}), (t^{(\ell)}, t)\} \cup \left\{(u_i^{(j)}, v_i^{(j)}) \mid i \in [2],\ j \in [\ell]\right\}
    \]  
    forms a perfect matching of size $2 + 2\ell$. Since there is no alternating cycle with respect to $M$, this matching is unique.
    
    For the third claim, let $M$ be the maximum matching described above. To construct a matching of size $\nu(G_\ell) - 1$, we remove the edge $(s, s^{(1)})$ from $M$.  
    Now consider an alternating path from $s^{(1)}$ to $t$, where for each $j \in [\ell]$, we select one of the two paths in $H^{(j)}$. Each such alternating path corresponds to a matching of size $\nu(G_\ell) - 1$, and all these matchings are distinct.  
    Hence, there are at least $2^\ell$ such matchings.
\end{proof}

\begin{proof}[Proof of \Cref{thm:lower-bound}]
    Consider a graph $G$ of $n$ vertices consisting of $G_\ell$ with $\ell = 1/(4\varepsilon)$ and $n-|V(G_\ell)|$ isolated vertices.
    Let $p$ be the probability that the we sample the unique perfect matching $M$ of size $\nu(G) = 2+1/(2\varepsilon)$.
    Then, we have
    \begin{align*}
        & p\nu(G) + (1-p)(\nu(G)-1) \geq (1-\varepsilon)\nu(G) \\
        \Rightarrow \quad & 
        p \geq 1 -\varepsilon\nu(G)
        = 1- \left(2\varepsilon+\frac{1}{2}\right) 
        = \frac{1}{2}- 2\varepsilon
        \geq \frac{1}{4},
    \end{align*}
    where the last inequality is by the assumption that $\varepsilon$ is sufficiently small.
    This implies that 
    \begin{align*}
        & \frac{\lambda^{\nu(G)}}{\lambda^{\nu(G)} + 2^\ell \lambda^{\nu(G)-1}} \geq p \geq \frac{1}{4} \\
        \Rightarrow \quad & \lambda^{\nu(G)} \geq \frac{1}{3} \cdot 2^\ell \lambda^{\nu(G)-1}\\
        \Rightarrow \quad & \lambda \geq \frac{1}{3} \cdot 2^\ell = 2^{\Omega(1/\varepsilon)}. 
    \end{align*}

\end{proof}

\bibliographystyle{siamplain}
\bibliography{main}

\appendix
\section{Sensitivity of Randomized Greedy upon Deleting a Vertex}\label{sec:flaw}

Given a graph $G = (V, E)$, the randomized greedy algorithm $A$ for the maximum matching problem works as follows. 
First, it chooses a random ordering $\pi$ over edges. 
Then starting with an empty matching $M$, it iteratively adds the $i$-th edge in the ordering $\pi$ to $M$ if it is not adjacent to any edge in $M$. 
Then, it is claimed in \cite[Theorem 12]{yoshida2021sensitivity} that, for any graph $G = (V, E)$ and a vertex $v \in V$, the sensitivity of $A$ upon deleting $v$ is one, that is, $\We{A(G), A(G - v)} \leq 1$, and this claim was crucially used to bound the sensitivity of their $(1-\epsilon)$-approximation algorithm for the maximum matching.
Unfortunately, their proof has a flaw, which we describe below.

First, we review their argument.
For a permutation $\pi$ over edges, let $M_\pi$ be the (deterministic) matching constructed by the randomized greedy algorithm when it uses $\pi$ as the ordering of edges.
For an edge $e \in E$, we call $\pi(e)$ the \emph{rank} of $e$, and let $I_\pi(e)$ be the set of edges sharing endpoints with $e$ with smaller rank, that is, $I_\pi(e) = \left\{e' \in N(e) \mid \pi(e') < \pi(e)\right\}$, where $N(e)$ is the set of edges that share an endpoint with $e$.
Note that the matching $M_\pi$ can be described by the following invariant:
An edge $e$ is in $M_\pi$ if and only if all of its neighbors $e' \in N(e) \cap I_\pi(e)$ are not in $M_\pi$.

Let $e_\pi \in E$ be the edge incident to $v$ with the smallest rank with respect to $\pi$.
We define $S_\pi \subseteq E$ to intuitively be the set of edges in $G$ that need to be changed to maintain the invariant. Formally, we set $S_{\pi,0} = \{e_\pi\}$ if $e_\pi \in M$ and $S_{\pi,0} = \emptyset$ otherwise. 
Then for $i > 0$, recursively set
\[
    S_{\pi,i} = \left\{e \in M_\pi \mid S_{\pi,i-1} \cap I_\pi(e) \neq \emptyset\right\} \cup \left\{e \not \in M_\pi \mid I_\pi(e)\cap M_\pi \subseteq \bigcup_{j=0}^{i-1}S_{\pi,j}\right\}
\]
We then define $S_\pi = \bigcup_i S_{\pi,i}$.
They showed that $\E[\pi]{|S_\pi|} \leq 1$, and claimed that it immediately implies $\We{A(G), A(G - v)} \leq 1$.
However, this argument ignores the case that some edge incident to $v$ with a higher rank than $e_\pi$ is in $M_\pi$ and deleting $v$ causes a cascade of changes to the matching through the edge.
A naive approach to address this issue is to define $S_\pi$ for every edge incident to $v$ and bound their sizes.
However, this only shows that $\We{A(G), A(G - v)} \leq \Delta(v)$, where $\Delta(v)$ is the degree of $v$, and it is not sufficient to obtain sensitivity independent of the maximum degree.

\section{Proofs in \Cref{sec:prelim}}
\label{sec:prelim-proof}
\renewcommand{\mp}[3]{\ifblank{#3}{\mu_{#1; \lambda, #2}^{\tau }}{\mu_{#1; \lambda, #2}^{\tau, #3} }}
\renewcommand{\mt}[2]{\ifblank{#2}{ \mathcal M^{\tau} (#1) }{ \mathcal M^{\tau, #2} (#1) }}

\begin{proof}[Proof of \Cref{lem:wd-sum}]
    This proof is an application of \Cref{lem:duality}, we have
    \begin{align*}
        \W{\mu, \mu'}
        &= \sup_{f\in L^1(\Omega)} \inner{f}{\mu - \mu'}
        = \sup_{f\in L^1(\Omega)} \sum_{i=1}^n\lambda_i\inner{f}{\mu_i - \mu'_i}
        \\&\le \sum_{i=1}^{n}
          \sup_{f_i\in L^1(\Omega)} \lambda_i\inner{f_i}{\mu_i - \mu'_i}
        = \sum_{i=1}^n \lambda_i\W{\mu_i, \mu'_i}.
    \end{align*}
\end{proof}

\begin{proof}[Proof of \Cref{lem:triangle}]
    Apply \Cref{lem:wd-sum} to the equation $\mu_1 -\mu_n = \sum_{i=1}^{n-1}(\mu_i - \mu_{i+1})$.
\end{proof}

\begin{proof}[Proof of \Cref{lem:pinning-deletion}]

For the first claim,
let $\mathcal M^{\tau, i \gets -}(G) = \set{M \in \mathcal M(G)\mid i\notin M, 
M \text{ satisfies }\tau}$.
Then for any $N\in \mathcal M^\tau(G[E'-i])=\mathcal M^\tau((G-i)[E'-i])$:
\begin{align*}
\mp {E'} G {i\gets -} (N)
&=
\left(\sum_{\substack{M\in \mathcal M^{\tau, i \gets -}(G)\\ M\cap (E'-i)= N}} \lambda^{|M|}\right)
\cdot
\left(\sum_{M\in \mathcal M^{\tau, i\gets -}(G)} \lambda^{|M|}\right)^{-1}.
\end{align*}
Notice that there is a one-to-one correspondence between the following sets:
\begin{align*}
    \mathcal M^{\tau, i \gets -}(G) \ni N &\leftrightarrow N \in  \mathcal M^\tau(G-i).
\end{align*}
Hence, we have
\begin{align*}
\mp{E'}{G}{i\gets -}(N)
&=
\left(\sum_{\substack{M\in \mathcal M^\tau(G-i)\\M\cap E' = N}} \lambda^{|M|}\right)
\cdot
\left(\sum_{\substack{M\in \mathcal M^\tau(G-i)}} \lambda^{|M|}\right)^{-1}\\
&=
\left(\sum_{\substack{M\in \mathcal M^\tau(G-i)\\M\cap (E'-i) = N}} \lambda^{|M|}\right)
\cdot
\left(\sum_{\substack{M\in \mathcal M^\tau(G-i)}} \lambda^{|M|}\right)^{-1}
=
\mp{E-i}{G-i}{}(N),
\end{align*}
where the second equation holds because 
\[
    \set{M\in \mathcal M^\tau(G-i)\mid M\cap E' = N} =\set{M\in \mathcal M^\tau(G-i)\mid M\cap (E'-i) = N}
\]
as $i$ is impossible to be matched.

The second claim is a corollary of the first one.
Notice that
$\mp{E-i}{G}{i\gets +} = \mp{E-i}{G}{i\gets +, F\gets -}$,
then we can apply the first equation to all edges in $F$. It gives
$\mp{E-i}{G}{i\gets +} = \mp{E\setminus F-i}{G\setminus F}{i\gets +}$
as desired.

Then we prove the third claim.
Let $\mathcal M^{\tau, i\gets +} = \set{M \in \mathcal M(G)\mid i\in M, M\text{ satisfies }\tau}$.
By the one-to-one correspondence
\begin{align*}
    \mathcal M^{\tau, i \gets +}(G) \ni N\cup\set i & \leftrightarrow N \in  M^\tau(G\setminus N(i)-i),
\end{align*}
we have
\begin{align*}
\mp{E-i} G {i\gets +}(N)
&=
\left(\sum_{\substack{M\in \mathcal M^{i \gets +}(G)\\ M\cap (E-i) = N}} \lambda^{|M|}\right)
\cdot
\left(\sum_{\substack{M\in \mathcal M^{i \gets +}(G)}} \lambda^{|M|}\right)^{-1}
\\&=
\left(\sum_{\substack{M\in \mathcal M(G\setminus N(i) - i)\\M\cap(E-i) = N}} \lambda^{|M'|+1}\right)
\cdot
\left(\sum_{\substack{M\in \mathcal M(G\setminus N(i) - i)}} \lambda^{|M|+1}\right)^{-1}
\\&=
\mp{E-N(i)-i} {G-N(i)-i} {}(N),
\end{align*}
where the second
equation holds because $\set{M\in \mathcal M^\tau(G\setminus N(i) - i)\mid M\cap (E-i) = N} =N$,
as all edges in $N(i)$ cannot be matched.
\end{proof}

\begin{proof}[Proof of \Cref{lem:expansion}]
Notice that the events $\set{i\gets - \mid i\in E(v)}\cup \set{N(i)\gets -}$ are mutually exclusive and 
collectively exhaustive. Then the lemma follows from the law of total probability.
\end{proof}

\section{Proof of \Cref{lem:DP}}\label{sec:dp}
\begin{proof}[Proof of \Cref{lem:DP}]
    We use the coupling constructed in the proof of \Cref{thm:bounded-degree-algorithm}.
    Recall that, in that proof, we split an edge into two pendant edges and applied inductive coupling to each one.
    For each pendant edge, at each inductive step, the probability that the discrepancy increasing is bounded by $\frac{\Delta\lambda}{1+\Delta\lambda}$.
    Hence, the probability that total discrepancy exceeds $\frac{c\log n}{\log\left(1 + \frac{1}{\lambda\Delta}\right)}$ is bounded by $n^{-c}$.
    Setting $\lambda = \Delta^{O(1/\eps)}$ completes the proof.
\end{proof}

\end{document}